\documentclass{article}

\usepackage{times}
\usepackage[pdftex]{graphicx}
\usepackage{amssymb,amsfonts,amsmath,amsthm}
\usepackage{float}
\usepackage{xcolor}
\usepackage{enumitem}

\newtheorem{theorem}{Theorem}
\newtheorem{proposition}{Proposition}
\newtheorem{lemma}{Lemma}
\newtheorem{ass}{Assumption}

\newtheorem{remark}{Remark}

\newcommand\EE {\mathbb E}

\newcommand\RR {\mathbb R}
\newcommand\PP {\mathbb P}
\newcommand\QQ {\mathbb Q}

\newcommand\ZZ {\mathbb Z}

\def\bone{\mathbf{1}}

\def\qed{\hskip6pt\vrule height6pt width5pt depth1pt}

\def\qed{\hskip 6pt\vrule height6pt width5pt depth1pt}

\newcommand{\ed}{\end{document}}
\newcommand{\be}{\begin{equation}}
\newcommand{\ee}{\end{equation}}
\newcommand{\bq}{\begin{eqnarray}}
\newcommand{\eq}{\end{eqnarray}}

\vspace{4in}

\newcommand{\cT}{\mathtt{T}}

\begin{document}

\title{A simple microstructural explanation of the concavity of price impact}
\author{Sergey~Nadtochiy\footnote{The author thanks M. Mouyebe for conducting a numerical experiment whose results inspired this article. He thanks S. Jaimungal for providing the market data used herein. The author also thanks J.-P. Bouchaud, P. Ustinov, K. Webster, and the anonymous referees, for the useful discussions and comments that helped him improve this paper significantly.}
\footnote{Partial support from the NSF CAREER grant 1855309 is acknowledged.}
\footnote{The data that support the findings of this study are available from the corresponding author upon reasonable request.}
\footnote{Address the correspondence to: Department of Applied Mathematics, Illinois Institute of Technology, 10 W. 32nd St., Chicago, IL 60616 (snadtochiy@iit.edu).}
}
\date{First draft: January 6, 2020
\\Current version: December 11, 2020
}

\maketitle

\begin{abstract}
This article provides a simple explanation of the asymptotic concavity of the price impact of a meta-order via the microstructural properties of the market. This explanation is made more precise by a model in which the local relationship between the order flow and the fundamental price (i.e. the local price impact) is linear, with a constant slope, which makes the model dynamically consistent. Nevertheless, the expected impact on midprice from a large sequence of co-directional trades is nonlinear and asymptotically concave. The main practical conclusion of the proposed explanation is that, throughout a meta-order, the volumes at the best bid and ask prices change (on average) in favor of the executor. This conclusion, in turn, relies on two more concrete predictions, one of which can be tested, at least for large-tick stocks, using publicly available market data.
\end{abstract}


\section{Introduction and main results}
\label{se:intro}

The term price impact often refers to the fact that, on average, trades move the asset price in their direction. However, this general observation has several more specific interpretations, and, as pointed out in \cite{Toth}, it is important to differentiate between various types of price impact. First, one may consider the local impact: i.e., the expected price change as a function of the volume of a single trade. The latter is only relevant in the markets where large trades occur often (e.g., OTC markets), but is not relevant, e.g., for the most popular stock exchanges. The second type of impact is the meta-order impact: i.e., the expected price change as a function of the volume of a sequence of co-directed trades (such sequence is called a meta-order). This type of impact is more relevant for public exchanges, where most participants willing to buy or sell a large quantity of the asset split this quantity into smaller pieces (child orders) and submit each of them separately. Within the meta-order impact, one can distinguish two sub-types, depending on whether the relative rate of the execution, or its total duration, is fixed. This paper studies the meta-order impact for a fixed execution rate (also known as the expected price trajectory).

\smallskip

There exist various empirical studies that confirm the concavity of meta-order impact (see, e.g., \cite{Almgren}, \cite{Bershova}, \cite{Bacry}, \cite{Zarinelli}, \cite{Moro}, \cite{Bucci}, \cite{CapponiCont}).\footnote{It is important to note that \cite{CapponiCont} both provides an empirical analysis of the price impact and proposes its explanation. However, the notion of price impact used in \cite{CapponiCont} is fundamentally different from the one used here and in other studies. Namely, \cite{CapponiCont} defines the price impact as the expected absolute value of the price change, which is of course dominated by the price volatility and is very different (qualitatively and quantitatively) from the expected value of the price change, which is more commonly used as the notion of price impact (and the latter notion is adopted herein).} They show that the expected price change as a function of traded volume is concave (see the left part of Figure \ref{fig:2}), and some even claim a specific power (in particular, square-root) dependence of the impact on the volume. These empirical results motivated the search for a sound theoretical explanation of the concavity of price impact. The following four paragraphs describe the main types of explanations available to date.

The first explanation is rather heuristic and has not been documented in the academic literature (to the author's best knowledge). It explains the concavity of price impact by the predictability of future prices and order flow. To understand this argument, notice that the concavity of price impact is equivalent to the statement that the marginal impact of the individual child orders decreases throughout a meta-order. Then, assuming, for example, that the price dynamics have a mean-reverting component, one can argue that the latter will slow down the deviation of the price from its initial level, causing smaller marginal impact of every subsequent trade.\footnote{Of course, this is only true if the initial price is on the right side of the level around which it mean-reverts. This is one of the flaws of this argument.} Similarly, if the market participants can predict (with some accuracy) the size of a meta-order, then, toward the end of its execution, they become less certain that the execution will continue, hence they limit their predatory actions agains the executor (i.e., moving their limit orders or front running), causing smaller marginal impact towards the end of the meta-order. One shortcoming of this type of arguments is, of course, is their reliance on the violation of the efficient market hypothesis, which states that prices are not predictable. Although, in practice, the latter holds only to a certain extent, typically, the inefficiencies in the market disappear over time, once they become known and are of a significant magnitude. Therefore, given the prolonged history of the empirical studies in this area, it is hard to imagine that the predictability of future prices and order flows would still exist on a level that would make a significant contribution to the price impact. Another downside of this explanation is the underlying assumption that the meta-order can be detected by other market participants instantaneously. Indeed, in practice, it takes some time for the predatory traders to detect a meta-order, implying that the very first few trades should cause smaller marginal impact than some of the subsequent ones, which would contradict the concavity of the price impact if the order flow predictability was the main driving force for the shape of the impact.

Another type of explanation is based on the game-theoretic models in which a market-maker provides liquidity for the executor. The concavity is explained, for example, via the shape of market-maker's preferences (\cite{Gabaix}) or via the distribution of the sizes of meta-orders (\cite{Farmer}). While such explanation is admissible for OTC markets, it is not clear whether the conclusions of such models extend to the order-driven markets (e.g., exchanges), where multiple heterogeneous liquidity providers with varying inventories (as they trade with other liquidity consumers) compete for clients' orders.

Additional evidence of the concavity (even square-root property) of price impact is obtained via the dimensional analysis in \cite{Kyle}, \cite{Pohl}. Indeed, if one assumes that the impact depends only on a certain group of factors, then, after few additional invariance assumptions, one can deduce that the square-root is the only function that produces the right unit of measurement for the impact. This is a rather powerful method, but, unfortunately, it sheds very little light on the mechanism that generates the price impact, which is an important part of explaining the phenomenon.

The latent limit order book model (LLOB), described, e.g. in \cite{Toth}, \cite{Donier}, and the model based on Hawkes processes, in \cite{Rosenbaum1}, provide two explanations of the concavity of price impact that are the closest to the one proposed herein. The LLOB model assumes the existence of a large number of potential buyers and sellers whose reservation prices for the asset follow a common stochastic process plus an idiosyncratic Brownian component independent across agents. When the reservation price of a buyer matches that of a seller, they trade and eliminate each other's orders. The price is defined as the point that separates the reservation prices of the buyers and sellers. In the large-population limit, the resulting distribution of reservation prices has a density that evolves deterministically, according to the heat equation, around the common price component. Then, assuming that a meta-order (of constant rate) can be represented by a source term in the heat equation, the authors of \cite{Donier} solve the resulting equation explicitly to recover the square-root price impact. Even though the model proposed in the present paper starts from a very different setting, the key feature of both models is that the liquidity improves on average throughout the meta-order, which implies that every subsequent child order makes smaller impact. This phenomenon of ``improving liquidity" is described in more detail further in this section and in Section \ref{subse:numExample}. It is also worth mentioning that, unlike the first two approaches discussed above, the explanation of the concavity of price impact provided by the LLOB and by the present model is purely mechanical and is consistent with the efficient market hypothesis: i.e., there is no need to assume any strategic behavior of other market participants, nor is it necessary to assume predictability of future prices and order flow.

The model of \cite{Rosenbaum1} (see also the references therein) starts by assuming that the order flow follows a self-exciting Hawkes process, which is defined by the associated kernel function, and that the price is driven only by the order flow (the latter is referred to as the no-arbitrage condition in \cite{Rosenbaum1} and in preceding papers). Then, assuming that the price impact has a non-trivial resilience (i.e., the price process obtains a trend in the opposite direction to the meta-order, right after the execution), the authors of \cite{Rosenbaum1} show that the only kernel that produces a well defined price impact in the infinite-activity regime (i.e., with small orders arriving at a high rate) is the power function. The latter leads to a concave power-type price impact. The argument of \cite{Rosenbaum1} is based on the properties of Hawkes process, one of the most important of which is the self-excitatory property. Even though the explanation of the concavity of price impact proposed herein is different, it is similar to \cite{Rosenbaum1} in that one of its main building blocks is the self-excitatory property of the fundamental price process, discussed further in this section. 

\medskip

The explanation of the concavity of price impact\footnote{It is worth mentioning that, strictly speaking, the present model only explains asymptotic concavity: i.e., that the marginal impact at the beginning of a long meta-order is higher than at its end. With a slight abuse of terminology, we refer to this property as concavity throughout the paper. However, the numerical experiments reveal that, for a wide range of parameters' values, the model produces globally concave price impact curves.} proposed herein has several advantages. First, as shown in the remainder of this section, the main results of this paper can be stated in plain language, without appealing to any sophisticated technical arguments. In this sense, the proposed explanation is model-free. Nevertheless, a specific modeling setting is described in Section \ref{se:math}, where, in particular, it is shown that the arguments presented further in this section to explain the concavity of price impact can co-exist in a reasonable mathematical model, and where the relevant predictions are proven rigorously (this is where the technical arguments are needed). The model of Section \ref{se:math} allows for multiple market participants, consuming and providing liquidity, which makes it well suited for order-driven markets. The setting of the model is ``bottom-up": i.e., its inputs have clear economic meaning and can be measured from market data. The proposed explanation does not require the predictability of the future prices or of the oder flow and, hence, is consistent with the efficient market hypothesis. Finally, a significant advantage of the proposed explanation is that one of its two most important predictions (which directly imply the concavity of price impact) can be tested on real market data \emph{without} the information about meta-orders (as the latter is notoriously difficult to obtain).

\smallskip

Most of the remainder of this section is devoted to the description of the proposed explanation of the concavity of price impact, and of the underlying assumptions, in plain language. The main assumption of this work is the existence of a real-valued process $X$, which we refer to as the fundamental price (this terminology comes from \cite{GaydukNadtochiy1}, \cite{GaydukNadtochiy2}, \cite{GaydukNadtochiy3}), and which has the meaning of a signal predicting the direction of the next trade. Namely, between the jump times of the best bid and ask prices, the process $X$ lies in the interval $[a,b]$ (which may change after the best bid or ask price changes), and the closer it is to $a$ (resp. $b$) the more likely it is that the next trade will be a sell (resp. buy). In addition, we assume that the best bid (resp. ask) price changes when and only when $X$ hits $a$ (resp. $b$).\footnote{This is similar in spirit to the model with uncertainty zones of \cite{Rosenbaum2}, although the latter model only describes the last transaction price, as opposed to the pair of the best bid and ask prices.} Another property that is required from $X$ is that every trade makes a positive linear local impact\footnote{The assumption of linearity of local impact is no loss of generality, as we ultimately consider a regime where every individual trade has infinitesimal size.} on $X$, with the coefficient $\alpha>0$. For example, a buy trade of size $\delta$ changes $X$ to $X+\alpha\delta$.
Any process $X$ that satisfies these properties is referred to as the fundamental price.

Although the proposed explanation works for any fundamental price $X$, it is convenient to think of a concrete example. Namely, measuring all prices in ticks and making the assumption that the bid-ask spread of the asset is always (or, in practice, almost always) equal to one (i.e., assuming that we study a large-tick asset), we conclude that the microprice $X$, defined as
\begin{equation}\label{eq.intro.microprice}
X_t = P^b + \frac{V^b_t}{V^b_t + V^a_t},
\end{equation}
where $P^b$ denotes the best bid price and $V^{b}_t$, $V^a_t$ denote the limit order volumes at the best bid and ask prices, respectively, satisfies all the properties of the fundamental price listed above. Indeed, the microprice always stays in the interval $[P^b,P^b+1]$, by the above definition, and it is known to possess the desired predictive power for the direction of the next trade (see, e.g., \cite{Stoikov}).
In addition, the large-tick property of the asset implies that the best ask (resp. bid) price changes when and only when $X$ hits $P^b+1$ (resp. $P^b$).\footnote{It is important to notice that this property fails for the assets whose bid-ask spreads vary significantly. This is why the microprice is not a good proxy for the fundamental price for such assets (see the also the discussion in Section \ref{subse:marketdata}).}
Finally, it is natural to expect that trades have a positive impact on the microprice: e.g., a part of this impact is purely mechanical, as any buy trade decreases $V^a$ and any sell trade decreases $V^b$.

For the sake of concreteness, all the subsequent analysis is presented under the assumption of a large-tick asset and with the microprice playing the role of the fundamental price $X$. Then, in particular, it is clear that the best bid and the best ask prices are given by the roundings $\lfloor X\rfloor$ and $\lceil X\rceil$, respectively, and that they change whenever $X$ hits an integer. Assuming that, between the trades, $X$ follows a symmetric random walk, we conclude that its global dynamics must have a force, or a drift term, that pushes $X$ away from the midprice. Indeed, if $X$ is slightly above the midprice, the next trade is more likely to be a buy, which pushes $X$ further up and makes the next buy even more likely, and so on. This can be viewed as the self-excitatory behavior of the fundamental price within the spread. For simplicity, let us focus on the dynamics of $X$ between two nearest integers -- i.e., we consider $X\,\texttt{mod}\,1$ ($X$ modulo one). As explained above, the process $X\,\texttt{mod}\,1$ is a random walk (on the unit circle) with a drift that pushes it away from $1/2$. The stationary distribution of such a process\footnote{Recall that the stationary distribution, in particular, is the distribution of the value of the process at any fixed time, provided the process has been running sufficiently long.} must have a U-shaped density (see the top left part of Figure \ref{fig:1}). Thus, before the execution of a meta-order begins, the distribution of the fundamental price modulo one has a U-shaped density. Next, let us analyze what happens toward the end of a meta-order. It is natural to assume that a meta-order, which is a sequence of co-directed trades, introduces an additional drift term in $X\,\texttt{mod}\,1$, of the same sign as the meta-order itself. If this additional drift is constant and sufficiently large, it is easy to deduce that the stationary distribution of the resulting process is uniform (think, e.g., of a Brownian motion with a very large drift, run on a circle). It is clear that the wings (i.e., the values at $0$ and $1$) of a U-shaped density are higher than the wings of a uniform density (the latter are equal to one). Interpolating heuristically between zero additional drift and the infinitely large drift, we conclude that the wings of the stationary density of $X\,\texttt{mod}\,1$ before a sufficiently long meta-order are higher than the wings of its density at the end of this meta-order, for any positive execution rate (compare the top left and the bottom right parts of Figure \ref{fig:1}). Now, the phenomenon of ``improving liquidity" becomes clear. Recalling that $X$ is the microprice (see \eqref{eq.intro.microprice}), one easily sees that the lower wings of the stationary density of $X\,\texttt{mod}\,1$ imply smaller probability of observing low liquidity (i.e., low volume of limit orders) at the best bid or ask. It only remains to connect the wings of the stationary density of $X\,\texttt{mod}\,1$ to the price impact directly. To this end, notice that the expected impact on the midprice of a buy trade of size $\delta$, submitted at time $t=0$, is given by 
$$
\EE \left( \lceil X_0\,\texttt{mod}\,1+\alpha\delta \rceil -1 \right),
$$
where $X_0\,\texttt{mod}\,1$ is a random variable whose density is given by the stationary density of $X\,\texttt{mod}\,1$.
As $\delta\downarrow0$, the leading order of the above expectation is given by the probability that $X_0\,\texttt{mod}\,1+\alpha\delta>1$. The latter, in turn, is proportional to the wings of the stationary density of $X\,\texttt{mod}\,1$. Thus, the expected marginal impact on the midprice is proportional to the wings of the stationary distribution of the fundamental price modulo one. As mentioned above, a meta-order introduces an additional drift term in the dynamics of the fundamental price, thus, switching the market into a different regime. In this new regime, $X\,\texttt{mod}\,1$ attains a new stationary density (provided the meta-order lasts long enough), whose wings are lower than the wings of the original stationary density. Repeating the above argument that connects the wings of the stationary density and the marginal impact, we conclude that the marginal impact at the end of the meta-order is lower than at the beginning, which implies the (asymptotic) concavity of price impact.

\smallskip

The above explanation needs an additional clarification, which leads to another assumption. Namely, the conclusion that the wings of the stationary distribution decrease during the meta-order relies on the assumption that the fundamental price obtains a constant drift during a meta-order, which is equivalent to the assumption that the meta-order is executed at a constant rate. However, a typical executor aims to hide her activity and trade according to (a version of) the Volume Weighted Average Price (VWAP) strategy. The latter states that the execution rate of a meta-order must be a constant fraction of the rate of the total traded volume in the market (see, e.g., \cite{VWAP}, \cite{JaimungalBook}). Thus, the drift that the fundamental price obtains during a meta-order is constant only if measured on the business clock -- i.e., using the total traded volume instead of the actual time. Therefore, all processes in the above discussion should be run on the business clock, and the conclusions of the previous paragraph require the additional assumption that the execution strategy is similar to VWAP.

\smallskip

Note that the explanation of the concavity of price impact presented above relies only on two predictions. The first one is the U-shape of the (global) stationary distribution of the fundamental price modulo one (run on a business clock). The second prediction is that, during a meta-order, the fundamental price (run on a business clock) obtains an additional constant drift term. While it is impossible to test the second prediction without the meta-order data (although this prediction appears to be self-evident), the first prediction is verified using real market data in Section \ref{subse:marketdata}.

\smallskip

The rest of the paper is organized as follows. 
The precise modeling assumptions and mathematical statements are given in Section \ref{se:math}. Sections \ref{subse:finact.jumpdiff.def} and \ref{subse:infAct.model} describe, respectively, the finite- and infinite-activity models for $X$, in particular, giving the precise form of its drift in terms of the input parameters of the model. The finite-activity model is easier to understand intuitively, while the infinite-activity model allows for more tractable representation of the price impact.
Sections \ref{subse:expectedImpact.def} and \ref{subse:LargeSample} are devoted to the definition and the computation of the price impact of a VWAP meta-order. These sections are complicated by the additional effort made in this paper to avoid treating the executor of a meta-order as an exogenous entity and to use a realistic definition of price impact, consistent with the way it would be computed in practice. In particular, it is shown in Section \ref{subse:expectedImpact.def} that the proposed setting allows for multiple agents following VWAP-type strategies, with overlapping execution intervals and with varying trade directions. The price impact is defined as the conditional expectation of the price change over the interval on which the first agent's trades sum up to a given volume (the size of the meta-order), given that all trades of this agent in this interval have the prescribed direction. The latter is done to mimic the computation of price impact as a sample average of the price change over the execution intervals.
Section \ref{subse:ImpactWings} shows that the marginal price impact is proportional to the wings of the stationary distribution of the fundamental price modulo one, run on a business clock. Propositions \ref{prop:Iq.lim} and \ref{prop:marginalimpact.at.zero} describe this connection, respectively, in the presence of a meta-order and in its absence. Section \ref{subse:concavity} proves that the aforementioned wings are lower in the presence of a meta-order (Theorem \ref{thm:main}), thus, establishing the asymptotic concavity of price impact. 
The numerical and empirical analysis is presented in Section \ref{se:empirical}.

\section{Mathematical analysis}
\label{se:math}

The roots of the proposed model go back to the game-theoretic setting of \cite{GaydukNadtochiy1}, \cite{GaydukNadtochiy2}, \cite{GaydukNadtochiy3}. However, the specific model proposed herein is natural enough and does not require any additional justification via equilibrium arguments. The core of the model is the assumption that the potential buyers and sellers arrive to the market one by one, having their own reservation prices (i.e., the ``fair" prices for the asset, in their view). If a reservation price of an agent is above (below) the current best ask (bid) price, a sell (buy) trade occurs. The reservation prices are not independent across the potential buyers and sellers: they have a common component $X$ and the idiosyncratic additive part, generated from a (symmetric around zero) distribution with c.d.f. $F$. It is shown in \cite{GaydukNadtochiy3} that, for reasonable values of the model parameters, the agents providing liquidity via limit orders set the equilibrium bid and ask prices exactly at $\lfloor X \rfloor$ and $\lceil X \rceil$, respectively, thus, making the model consistent with the setting described in Section \ref{se:intro}. The details of the model are presented in the remainder of this section, with the main result (the asymptotic concavity of price impact) stated in Theorem \ref{thm:main}. It is worth mentioning that a part of this section contains the description of two models: with finite and infinite trading activity. The former is easier to interpret from the economic (or practical) point of view. The latter provides more tractable expressions for the target quantities. We show that the two are consistent and, ultimately, focus our attention on the infinite-activity model.

\subsection{Finite-activity jump-diffusion model}
\label{subse:finact.jumpdiff.def}

The potential buyers/sellers arrive according to a Poisson process with jump times $\{S_i\}$ and intensity $\lambda$. The reservation price of the $i$th buyer/seller is
$$
p^0_{S_i} = \tilde X_{S_i-} + \xi_i,
$$
where $\tilde X_{S_i-}:=\lim_{t\uparrow S_i} \tilde X_t$, $\{\xi_i\}$ are i.i.d. random variables, with c.d.f. $F$, and $\tilde X$ evolves according to
$$
\tilde X_t = X_0 + \alpha\sum_{S_i\leq t} \Delta V_{S_i} + \sigma(\tilde X_t) \tilde B_t\,\,\text{mod}\,\,1,
$$
with $\tilde B$ being a Brownian motion independent of $\{S_i,\xi_i\}$, and with
$$
\Delta V_{S_i} := \delta \bone_{\{\xi_i\geq \lceil \tilde X_{S_i-} \rceil - \tilde X_{S_i}\}} - \delta \bone_{\{\xi_i\leq \lfloor X_{S_i-}\rfloor - X_{S_i} \}}.
$$
Throughout the rest of the paper, we make the following standing assumptions on $F$ and $\sigma$.
\begin{itemize}
\item $F\in C^{1+\epsilon}([-1,1])$, for some $\epsilon\in(0,1)$, and $F'$ is symmetric around $x=0$ in this range.
\item $\inf_{x\in[0,1]}(F(x-1)+F(-x))>0$.
\item $\sigma$ has period one and is symmetric around $x=1/2$.
\item $\inf_{x\in[0,1]}\sigma(x)>0$.
\item $\sigma\in C^{1+\epsilon}([0,1])$.
\end{itemize}
In the above, we use the standard notation $C^{n+\epsilon}([a,b])$ to denote the space of real-valued functions on $[a,b]$ that are $n$ times continuously differentiable, with $\epsilon$-H\"older derivatives of order $n$.

Denote by $M$ a Poisson random measure with the compensator
$$
\mu(dt,dx) = \lambda dt\otimes (\PP\circ \xi^{-1}_i)(dx)
= \lambda dt\otimes dF(x),
$$
independent of $\tilde B$. Then, the fundamental price and the order flow are described by the following system:
\begin{equation}\label{eq.tildeX.def}
\left\{
\begin{array}{l}
{\tilde X_t = X_0 + \alpha (N^+_t - N^-_t) + \sigma(\tilde X_t) \tilde B_t,}\\
{N^+_t = \delta\int_0^t \int_{\RR} \bone_{\{x\geq \beta^+(\tilde X_{u-})\}} M(du,dx),}\\
{N^-_t = \delta\int_0^t \int_{\RR} \bone_{\{x\leq \beta^-(\tilde X_{u-})\}} M(du,dx),}
\end{array}
\right.
\end{equation}
where
$$
\beta^+(x) := \lceil x \rceil - x,\quad \beta^-(x) := \lfloor x\rfloor - x.
$$
Note that
$$
N^+_t = \sum_{S_i\leq t} \max(\Delta V_{S_i},0),\quad N^-_t=\sum_{S_i\leq t} \max(-\Delta V_{S_i},0).
$$
The input to the model is $(\alpha,\sigma,F,\lambda,\delta)$.

\subsection{Infinite-activity diffusion model}
\label{subse:infAct.model}

For analytic tractability, it is convenient to consider an infinite-activity limit of the model \eqref{eq.tildeX.def}, as $\lambda\rightarrow\infty$. In order to avoid the explosion of total order flow, we need to assume that $\delta\rightarrow0$, so that $\lambda\delta \rightarrow\gamma$, with some constant $\gamma>0$. For simplicity, we assume that $\delta=\gamma/\lambda$.
Notice that
$$
dN_t := d (N^+_t - N^-_t) = \lambda\delta \left(F(-\beta^+(\tilde X_t)) - F(\beta^-(\tilde X_t))\right) dt + dZ_t,
$$
where $Z$ is a martingale.
Then,
$$
\langle Z\rangle_t = \int_0^t \lambda\delta^2 \left(F(-\beta^+(\tilde X_s)) + F(\beta^-(\tilde X_s))\right) ds
\sim \delta\gamma\int_0^t \left(F(-\beta^+(\tilde X_s)) + F(\beta^-(\tilde X_s))\right) ds\rightarrow0.
$$ 
Thus, we expect $\tilde X$ to converge to $X$ which is the solution of
\begin{equation}\label{eq.DiffModel.X.balanced}
dX_t = \alpha \gamma \left(F(-\beta^+(X_t)) - F(\beta^-(X_t))\right) dt + \sigma(X_t) dB_t,
\end{equation}
equipped with the same initial condition as $\tilde X$.
We do not make this statement precise, as we will in fact need the convergence of time-changed processes, established in Lemma \ref{le:NY.conv}.

\smallskip

It is important to notice that the drift of $X$ in \eqref{eq.DiffModel.X.balanced} is increasing on $(n,n+1)$, negative in $(n,n+1/2)$, and positive in $(n+1/2,n+1)$, for any integer $n$. Indeed, for $x\in(n,n+1)$, the expression
$$
F(-\beta^+(x)) - F(\beta^-(x)) = F(x-n-1) - F(n - x)
$$
is increasing in $x$, as the c.d.f. $F$ is an increasing function, and the above expression is equal to zero at $x=n+1/2$.
These observations imply that the drift of $X$ pushes it away from the midprice, consistent with the conclusions of Section \ref{se:intro}.


\subsection{Expected impact on midprice by a VWAP strategy}
\label{subse:expectedImpact.def}

In this subsection, we define the expected price impact of a VWAP execution strategy and find its convenient representation (equation \eqref{eq.tildeI.rep.1}) in the finite-activity model described above. It is assumed that the impact is computed from a sample of consecutive past trades of the agent of the total size $L$. Then, we define the expected price impact in the infinite-activity model as a corresponding limit of the finite-activity impact (see \eqref{eq.I.def}) and show that it has a natural interpretation in terms of the infinite-activity model itself (Proposition \ref{prop:I.via.tildeI}).

\subsubsection{Finite activity model}

Recall that the process $(N_t= N^+_t - N^-_t)$ represents the total order flow in the market. In the finite activity model \eqref{eq.tildeX.def}, we have
\begin{align}
&\tilde X_t = X_0 + \alpha N_t + \int_0^t \sigma(\tilde X_u) d\tilde B_u,\label{eq.tildeX.def.2}\\
&N_t = \delta\int_0^t \int_{\RR} \left(\bone_{\{x\geq \beta^+(\tilde X_{u-})\}}- \bone_{\{x\leq \beta^-(\tilde X_{u-})\}}\right) M(du,dx),\label{eq.N.def.2}
\end{align}
where $\tilde B$ is a Brownian motion and $M$ is an independent Poisson random measure (see, e.g., \cite[Chapter II]{JacodShiryaev}) with the compensator
$$
\mu(dt,dx) = \lambda dt\otimes dF(x).
$$

The total order flow is a sum of the order flows of $K$ individual market participants (agents):
\begin{equation}\label{eq.N.via.Nj}
N = \sum_{j=1}^K N^j.
\end{equation}
We assume that the agents follow VWAP-type strategies. Namely, the $j$th agent has the order flow
\begin{equation}\label{eq.Nj.def}
N^j_t = \delta\int_0^t \int_{\RR} \left(\bone_{\{x\geq \beta^+(\tilde X_{u-})\}} + \bone_{\{x\leq \beta^-(\tilde X_{u-})\}}\right) \zeta^j(\tilde X_u-,u) M^j(du,dx),\quad j=1,\ldots,K,
\end{equation}
where $\{M^j\}$ are independent Poisson random measures with the compensators
$$
\mu^j(dt,dx) = \lambda^j dt\otimes dF(x),\quad j=1,\ldots,K,
\quad
\sum_{j=1}^K \lambda^j = \lambda,
$$
and $(\omega,x,u)\mapsto\zeta^j(x,u)\in\{\pm1\}$ is a random field, defined for all $(x,u)\in\RR\times\RR_+$, s.t.
\begin{itemize}
\item $\{\zeta^j(\cdot,u)\}$ are independent across $j=1,\ldots,K$, across $u\geq0$, and independent of $(\{M^j\},\tilde B)$,
\item for each $j=1,\ldots,K$ and $(x,u)\in\RR\times\RR_+$, we have
\begin{equation}\label{eq.sec2.Pzetaj.def}
\PP(\zeta^j(x,u)=1) = \frac{F(-\beta^+(x))}{F(-\beta^+(x)) + F(\beta^-(x))}.
\end{equation}
\end{itemize}

The random fields $\{\zeta^j\}$ represent the heterogeneity of the agents in terms of their trading styles. For example, at any given moment in time, one agent may buy (i.e., $\zeta^j=1$), while another one may sell (i.e., $\zeta^j=-1$). In general, we allow the agents' trading styles to depend on the fundamental price and on their idiosyncratic sources of randomness. The assumptions we make on $\{\zeta^j\}$ are not the most general, but they ensure that the model is consistent: i.e., the total order flow $N=\sum_{j=1}^K N^j$ satisfies \eqref{eq.N.def.2}, with an appropriately chosen Poisson random measure $M$ having the prescribed compensator and being independent of $\tilde B$. The latter observation is made precise in the appendix.



To see why we call the agents' strategies VWAP-type,notice that, in the present model, the total traded volume in the market at time $t$ is given by
$$
\tilde V_t = \delta\int_0^t \int_{\RR} \left(\bone_{\{x\geq \beta^+(\tilde X_{u-})\}} + \bone_{\{x\leq \beta^-(\tilde X_{u-})\}}\right) M(du,dx).
$$
Comparing the above with \eqref{eq.Nj.def}, we conclude that the $j$th agent in the proposed model trades with the rate that is $\lambda^j/\lambda$ fraction of the overall rate, which makes it similar to VWAP. However, unlike the classical VWAP strategy, where the trades are made in the same direction, our agents trade in different directions. It is also worth noting the difference between $N$ (in \eqref{eq.N.def.2}) and $\tilde V$ given above: the former denotes the order flow process, to which the buy market orders contribute with the positive sign and the sell orders contribute with the negative sign, while $\tilde V$ denotes the traded volume process, to which all trades contribute with the positive sign.

\medskip

Let us now assume that the first agent aims to compute the expected impact of a sequence of her buy trades (referred to as the execution interval) on the midprice.\footnote{We only consider the execution intervals of buy trades, as the case of sell trades is analogous.}
Recalling \eqref{eq.Nj.def}, we conclude that any such execution interval can be characterized by the condition $\zeta^1=1$.
In order to compute this expected impact in practice, the agent (i) finds the past execution intervals in a sample of past $L$ trades, (ii) records the changes in the quoted ask price over each interval, and (iii) computes the sample average of these changes. Mathematically, this process corresponds to computing
$$
\tilde I_{L}(Q,\delta,\lambda,\lambda^1):=\EE \left(\lceil \tilde X_{\tau_1}\rceil - \lceil \tilde X_{\tau_0}\rceil\,\vert\, \zeta^1(\cdot,t)=1,\,t\in(\tau_0,\tau_1]\right),
$$
where
the constants $Q>0$ and $L>0$ denote, respectively, the (fixed) total amount of the asset purchased by the first agent over each randomly selected execution interval and the size of the sample (measured in trades of the first agent) from which the execution intervals are selected. The random times $\tau_0$ and $\tau_1$ denote, respectively, the beginning and the end of a randomly chosen execution interval:
$$
\tau_0:= \cT(\tilde V^1,\eta),
\quad \tau_1:=\cT(\tilde V^1, \tilde V^1_{\tau_0}+Q),
\quad \eta\sim U(0,L),
$$
\begin{equation}\label{eq.tildeV1.def}
\tilde V^1_t := \delta\int_0^t \int_{\RR} \left(\bone_{\{x\geq \beta^+(\tilde X_{u-})\}} + \bone_{\{x\leq \beta^-(\tilde X_{u-})\}}\right) M^1(du,dx),
\end{equation}
where $U(0,L)$ denotes the uniform distribution on $[0,L]$, we assume that $\eta$ is independent of everything else, and we introduce the hitting time (or inverse) functional
$$
\cT(Z,h):=\inf\{t\geq0:\, Z_t>h\},
$$
for any process $Z$ on $[0,\infty)$ and any level $h\in\RR$. Note that $\tilde V^1$ defined above represents the first agent's traded volume process.
Thus, the above choice of $\tau_0$ corresponds to the assumption that each trade of the agent is equally likely to be the first trade of a VWAP meta-order.

\medskip

Next, we notice that \eqref{eq.tildeX.def.2}--\eqref{eq.Nj.def} imply
\begin{align}
&\tilde X_t = X_0 + \alpha\delta \int_0^t\int_\RR\sum_{j=1}^K \left(\bone_{\{x\geq \beta^+(\tilde X_{u-})\}} + \bone_{\{x\leq \beta^-(\tilde X_{u-})\}}\right) \zeta^j(\tilde X_{u-},u) M^j(du,dx) + \int_0^t \sigma(\tilde X_u) d\tilde B_u.\label{eq.rev.sec2.tilde X.rep}
\end{align}
Using \eqref{eq.rev.sec2.tilde X.rep} and \eqref{eq.tildeV1.def}, as well as the definition of the random field $\zeta^1$, we easily deduce that $\tilde X_{\tau_0}$ is independent of the values of $\zeta^1$ on the random time interval $(\tau_0,\tau_1]$.
Recall that $\tau_1-\tau_0$ is a function of the path $Z_t:=\tilde V_{\tau_0+t}-\tilde V_{\tau_0}$ for $t\in[0,\cT(Z,Q)]$: indeed, 
$$
\tau_1-\tau_0 = \cT(\tilde V^1, \tilde V^1_{\tau_0}+Q) - \tau_0 = \cT(Z,Q).
$$
Thus, $\tilde X_{\tau_1}$ is a function of the path of $(Z_t, \tilde X_{\tau_0+t})$ for $t\in[0,\cT(Z,Q)]$. For convenience, we denote 
$$
\tilde X_{\tau_1}=:g\left(Z_t, \tilde X_{\tau_0+t},\,\,t\in[0,\cT(Z,Q)]\right).
$$
Using \eqref{eq.rev.sec2.tilde X.rep} and \eqref{eq.tildeV1.def} again, we deduce that the conditional distribution of the path of $(Z_t, \tilde X_{\tau_0+t})$ for $t\in[0,\cT(Z,Q)]$, given $\tilde X_{\tau_0}=x$ and $\zeta^j(\cdot,t)=1$ for $t\in(\tau_0,\tau_1]$, is the same as the distribution of the path of $(\tilde N_t(x), \tilde Y_t(x)$, for $t\in[0,\cT(\tilde N(x),Q)]$, where,
\begin{equation}\label{eq.tildeN.def}
\tilde N_t(x) := \delta \int_0^t \int_\RR \left(\bone_{\{z\geq \beta^+(\tilde Y_{u-}(x))\}} + \bone_{\{z\leq \beta^-(\tilde Y_{u-}(x))\}}\right) \tilde M(du,dz),
\end{equation}
\begin{equation}\label{eq.tildeY.def}
\tilde Y_t(x) := x + \alpha \tilde N_t(x) + \alpha\delta\int_0^t \int_{\RR} \left(\bone_{\{z\geq \beta^+(\tilde Y_{u-}(x))\}}- \bone_{\{z\leq \beta^-(\tilde Y_{u-}(x))\}}\right) \hat M(du,dz) + \int_0^t \sigma(\tilde Y_u(x)) d\tilde W_u,
\end{equation}
for any $x\in\RR$, with $\tilde W$ being a Brownian motion, and with $\tilde M$ and $\hat M$ being independent Poisson random measures with the respective compensators
$$
\tilde\nu(dt,dx) = \lambda^1 dt\otimes dF(x),
\quad\hat\nu(dt,dx) = (\lambda-\lambda^1) dt\otimes dF(x).
$$
By possibly extending the probability space, we assume that $(\tilde Y(x),\tilde N(x))$ are constructed on the same probability space as $(\tilde X,\tilde V^1,\eta)$ and are independent of the latter.
The above observations imply
\begin{align*}
&\tilde I_{L}(Q,\delta,\lambda,\lambda^1)=\EE \left(\lceil \tilde X_{\tau_1}\rceil - \lceil \tilde X_{\tau_0}\rceil\,\vert\, \zeta^1(\cdot,t)=1,\,t\in(\tau_0,\tau_1]\right)= - \EE \lceil \tilde X_{\tau_0}\rceil\\
&+ \EE\left( \EE \left(\lceil g(Z_t, \tilde X_{\tau_0+t},\,\,t\in[0,\cT(Z,Q)]) \rceil\,\vert\,\tilde X_{\tau_0},\, \zeta^1(\cdot,t)=1,\,t\in[\tau_0,\tau_1]\right) \,\vert\,\zeta^1(\cdot,t)=1,\,t\in(\tau_0,\tau_1] \right)\\
&= \EE \left[\left( \EE \lceil g(\tilde N_t(x), \tilde Y_{t}(x),\,\,t\in[0,\cT(\tilde N(x),Q)])\right)_{x=\tilde X_{\tau_0}} \rceil\,\vert\, \zeta^1(\cdot,t)=1,\,t\in[\tau_0,\tau_1]\right] - \EE \lceil \tilde X_{\tau_0}\rceil\\
&= \EE \lceil \tilde Y_{\cT(\tilde N(x),Q)}(\tilde X_{\tau_0}) \rceil - \EE \lceil \tilde X_{\tau_0}\rceil,
\end{align*}
where we recalled the meaning of function $g$ to obtain the las equality.

Introducing the time-changed processes
\begin{align*}
& \bar X_v := \tilde X_{\cT(\tilde V^1,v)},\quad  \bar Y_v(x) := \tilde Y_{\cT(\tilde N(x),v)}(x),
\end{align*}
we summarize the above result as 
\begin{equation}\label{eq.tildeI.rep.1}
\tilde I_{L}(Q,\delta,\lambda,\lambda^1)=\EE \left(\lceil \bar Y_Q(\bar X_\eta)\rceil - \lceil \bar X_\eta\rceil\right),
\end{equation}
where $(\tilde V^1,\tilde X)$ are defined in \eqref{eq.tildeV1.def}, \eqref{eq.rev.sec2.tilde X.rep}, and $(\tilde Y^x,\tilde N^x)$ are defined in \eqref{eq.tildeN.def}, \eqref{eq.tildeY.def}.

\smallskip

The process $\tilde N$ represents the order flow of the first agent during her (buy) execution intervals.
The interpretation of $\tilde Y$ is also clear: it is a model for the fundamental price during the (buy) execution intervals of the first agent. Since the order flow is biased upwards in such intervals, the integrand in \eqref{eq.tildeN.def} has a sum of two (mutually exclusive) indicators, and the resulting nondecreasing process $\tilde N$ creates an upward trend in the dynamics of $\tilde Y$. Note also that the processes $\bar X$ and $\bar Y(x)$ run on the business clock, hence, their indices are measured in traded volume.

\subsubsection{Infinite-activity model}

Note that the expected impact on midprice, given by \eqref{eq.tildeI.rep.1}, depends only on the distributions of $\bar X$ and $\bar Y$, defined in \eqref{eq.tildeI.rep.1}.
Let us fix $\gamma>0$ and $\theta\in(0,1]$, and consider the limit of the expected impact:
\begin{equation}\label{eq.I.def}
I_{L}(Q,\theta):=\lim_{\lambda\rightarrow\infty} \tilde I_L(Q,\gamma/\lambda,\lambda,\theta\lambda),
\end{equation}
provided it is well defined. The restriction $\delta=\gamma/\lambda$ is discussed at the beginning of Subsection \ref{subse:infAct.model}. The condition $\lambda^1=\theta\lambda$ is equivalent to the assumption that the market participation rate of the first agent is fixed as we vary $\lambda$.

\smallskip

It turns out that $\bar X$ converges (weakly) to $\hat X$, given by 
\begin{equation}\label{eq.Xhat.def}
\hat X_t = X_0 + \int_0^t \hat\mu_0(\theta,\hat X_u) du
+ \int_0^t \hat\sigma(\theta,\hat X_u) d\hat B_u,
\end{equation}
with a Brownian motion $\hat B$ and with
$$
\hat\mu_0(\theta,y):= \alpha\frac{F(-\beta^+(y)) - F(\beta^-(y))}{\theta\left(F(-\beta^+(y)) + F(\beta^-(y))\right)},
\quad \hat\sigma(\theta,y):= \frac{\sigma(y)}{\sqrt{\theta\gamma (F(-\beta^+(y))+F(\beta^-(y)))}},
$$
while $\bar Y(x)$ converges (weakly) to $\hat Y(x)$, given by
\begin{equation}\label{eq.Y.def}
\hat Y_t(x) = x + \int_0^t \hat\mu_1(\theta,\hat Y_u(x)) du
+ \int_0^t \hat\sigma(\theta,\hat Y_u(x)) d\hat W_u,
\end{equation}
with a Brownian motion $\hat W$ and with
$$
\hat\mu_1(\theta,y):= \alpha\frac{2\theta F(\beta^-(y)) + F(-\beta^+(y)) - F(\beta^-(y))}{\theta\left(F(-\beta^+(y)) + F(\beta^-(y))\right)}.
$$

To make the above statement precise, we view $\bar X$, $\bar Y(x)$ as random elements with values in the Skorokhod space $D([0,\infty))$, and $\hat X$, $\hat Y(x)$ as random elements with values in $C([0,\infty))$. Note also that the laws of the processes $\hat X$ and $\hat Y(x)$ are uniquely determined by \eqref{eq.Xhat.def} and \eqref{eq.Y.def}, which can be easily seen by applying the scale function transformation and reducing these SDEs to the ones with no drifts and with Lipschitz diffusion coefficients.

\begin{lemma}\label{le:NY.conv}
As $\lambda\rightarrow\infty$, for any $x\in\RR$, we have:
$$
\PP\circ \bar X^{-1} \rightarrow \PP\circ \hat X^{-1},
\quad \PP\circ (\bar Y(\bar X_\eta))^{-1} \rightarrow \PP\circ (\hat Y(\hat X_\eta))^{-1},
$$
where the convergence is in weak topology induced by the $C([0,\infty))$-seminorms.
\end{lemma}
\begin{proof}
W.l.o.g., we only prove the convergence of $\bar X$. Recall that the latter is a function of two processes, $\tilde X$ and $\tilde V^1$.
First, we prove the $C$-tightness of the joint law of $(\tilde X, \tilde V^1)$ over $\lambda\rightarrow\infty$ on a finite time interval $[0,T]$. To prove the latter, it suffices to show (i) that the absolute values of the two processes are bounded in probability, uniformly over $\lambda$, and (ii) that
$$
\forall\,\varepsilon>0,\quad \lim_{\varepsilon'\rightarrow0}\limsup_{\lambda\rightarrow\infty} \PP\left(\sup_{t,s\in[0,T],\,|t-s|\leq\varepsilon'} |Z_t-Z_s|>\varepsilon\right) = 0,
$$
for $Z=\tilde X, \tilde V^1$. Both (i) and (ii) follow from Chebyshev's inequality and the estimate
$$
\EE \sup_{u\in[t,s]}|Z_u-Z_s| \leq C \lambda_1\delta |t-s|,
$$
where $C$ is a constant and we recall $\lambda_1=\theta\lambda$ and $\delta=\gamma/\lambda$.

Next, we consider any limit point $\Lambda$ (a probability measure on $(C([0,T]))^2$) of the family $\{\PP\circ(\tilde X, \tilde V^1)^{-1}\}_\lambda$, with the associated sequence $\{\lambda_n\rightarrow\infty\}$. In particular, $(\tilde X^n, \tilde V^{1,n})\rightarrow (\check X, \check V^{1})$ in weak topology induced by the $C$-norm. Let us describe the dynamics of $(\check X, \check V^{1})$.
It is easy to see that, for any $f\in C_b(\RR)$,
$$
\EE \sup_{t\in[0,T]} \left| f(\check V^1_t) - \gamma\theta \int_0^t f'(\check V^1_s) \left(F(-\beta^+(\check X_s)) + F(\beta^-(\check X_s))\right) ds \right|
$$
$$
=\lim_{n\rightarrow\infty}\EE \sup_{t\in[0,T]} \left| f(\tilde V^{1,n}_t)- \gamma\theta \int_0^t f'(\tilde V^{1,n}_s) \left(F(-\beta^+(\tilde X^n_s)) + F(\beta^-(\tilde X^n_s))\right) ds \right|=0.
$$
Choosing an appropriate sequence of $f$ approximating the identity, we deduce from the above that
\begin{equation}\label{eq.checkV}
\check V^1_t = \gamma\theta \int_0^t \left(F(-\beta^+(\check X_s)) + F(\beta^-(\check X_s))\right) ds,\quad t\in[0,T].
\end{equation}
Similarly, for any $0<t<s\leq T$, $f\in C_b(\RR)$, $0\leq t_1<\cdots<t_k\leq t$, $g\in C_b(\RR^{2k})$,
$$
\EE \left[g\left(\check X_{t_1}, \check V^{1}_{t_1},\ldots,\check X_{t_k}, \check V^{1}_{t_k}\right) \left( f(\check X_s) - f(\check X_t) - \alpha\gamma \int_t^s f'(\check X_u) \left(F(-\beta^+(\check X_u)) - F(\beta^-(\check X_u))\right) du\right)\right]
$$
$$
=\lim_{n\rightarrow\infty}\EE \left[g\left(\tilde X^n_{t_1}, \tilde V^{1,n}_{t_1},\ldots,\tilde X^n_{t_k}, \tilde V^{1,n}_{t_k}\right) \left( f(\tilde X^{n}_s) - f(\tilde X^{n}_t) \right.\right.
$$
$$
\phantom{??????????????????}
\left.\left.- \alpha\gamma \int_t^s f'(\tilde X^{n}_u) \left(F(-\beta^+(\tilde X^n_u)) - F(\beta^-(\tilde X^n_u))\right) du\right)\right]=0.
$$
From the above, we deduce that the process $M_t:= \check X_t - \gamma \int_0^t \left(F(-\beta^+(\check X_u)) - F(\beta^-(\check X_u))\right) du$, defined on the canonical space $(C([0,T]))^2$, is a continuous martingale under $\Lambda$ and, therefore, is given by a Brownian integral. Using the test function, as in the above, we easily deduce that $d\langle M\rangle_t=\sigma^2(\check X_t)dt$ a.s. under $\Lambda$. Thus we have shown that $\check X$ can be written as
\begin{equation}\label{eq.checkX}
\check X_t = X_0 + \alpha\gamma \int_0^t \left(F(-\beta^+(\check X_u)) - F(\beta^-(\check X_u))\right) du + \int_0^t \sigma(\check X_u) d\check B_u,\quad t\in[0,T],
\end{equation}
where $\check B$ is a Brownian motion under $\Lambda$. As the law of $(\check X, \check V^1)$ is uniquely determined by \eqref{eq.checkV} and \eqref{eq.checkX}, the convergence of $(\tilde X^n,\tilde V^{1,n})$ holds along any sequence $\{\lambda_n\rightarrow\infty\}$.

To conclude the proof, we notice that there exists $\varepsilon>0$, s.t., $\Lambda$-a.s., $\check V^1_\cdot \in K^\varepsilon$, with $K^\varepsilon:=\{f\in C([0,T]):\, f(t)-f(s)\geq \varepsilon(t-s),\,\,\forall\, 0\leq s < t \leq T\}$. It is easy to see that the mapping $(f,g)\mapsto f\circ g^{-1}$ is a continuous mapping from $C([0,T])\times K^\varepsilon$ into $C([0,T])$. Thus, using the Skorokhod's representation theorem and the portmanteau theorem, we conclude that, along any $\{\lambda_n\rightarrow\infty\}$,
$$
\bar X^n_\cdot := \tilde X^n_{(\tilde V^{1,n}_\cdot)^{-1}}\rightarrow \check X_{(\check V^{1}_\cdot)^{-1}}=:\hat X_\cdot,
$$ 
with the inverse being defined as a right-continuous function and with the convergence being in weak topology induced by the $C$-norm. Using \eqref{eq.checkV} and \eqref{eq.checkX}, we easily show that $\hat X$ satisfies \eqref{eq.Xhat.def}. Recalling that the solution to \eqref{eq.Xhat.def} is unique in law, we complete the proof of the lemma.
\qed
\end{proof}

\medskip

\begin{remark}
The proof of Lemma \ref{le:NY.conv} also shows that, in the infinite-activity model, during an execution interval of the first agent, her order flow is given by
$$
\theta \gamma \int_0^\cdot \left(F(-\beta^+(Y_u)) + F(\beta^-(Y_u))\right) du,
$$
where $Y$ represents the dynamics of the fundamental price in such intervals (it is the limit of $\tilde Y$). In addition, during any such interval, the total traded volume in the market is given by 
\begin{equation}\label{eq.totVolume.infAct}
\gamma \int_0^\cdot \left(F(-\beta^+(Y_u)) + F(\beta^-(Y_u))\right) du.
\end{equation}
Thus, in the infinite-activity limit, the first agent still uses a VWAP strategy, with the participation rate $\theta$.
\end{remark}

In view of Lemma \ref{le:NY.conv}, it is natural to expect that $I_{L}(Q,\gamma,\theta)$, given by \eqref{eq.I.def}, can be computed by replacing $(\bar X,\bar Y)$ by $(\hat X,\hat Y)$ in \eqref{eq.tildeI.rep.1}.

\begin{proposition}\label{prop:I.via.tildeI}
For any $L>0$, $Q>0$, and $\theta\in(0,1]$, the limit in \eqref{eq.I.def} is well defined, and we have
\begin{equation}\label{eq.I.rep.1}
I_{L}(Q,\theta)=\EE \left(\lceil \hat Y_Q(\hat X_{\eta})\rceil - \lceil \hat X_{\eta}\rceil\right),
\end{equation}
with $\eta\sim U(0,L)$ independent of $(\hat X,\hat Y)$.
\end{proposition}
\begin{proof}
The proof follows from Lemma \ref{le:NY.conv}, the portmanteau theorem, and the fact that neither $\hat Y_Q(\hat X_{\eta})$ nor $\hat X_{\eta}$ have atoms.
\qed
\end{proof}

\smallskip

In the remainder of the paper, we stay in the setting of the infinite-activity model.

\subsection{Large-sample limit}
\label{subse:LargeSample}

In this subsection, we consider the limit of the expected price impact of a VWAP meta-order in the infinite-activity model, as the sample size $L$ over which the impact is estimated increases to infinity (see \eqref{eq.Iinf.def}). We then express the resulting (infinite-activity and infinite-sample impact) through the stationary density of the fundamental price run on the business clock (Proposition \ref{prop:impact.inf.L}). 

\smallskip

Recall that $\eta\sim U(0,L)$, where $L$ represents the length of the data sample from which the execution intervals are collected.\footnote{Note that the length of an individual execution interval (measured in trades of the first agent) may be much smaller than $L$: i.e., at this stage, we do not make the assumption that the execution intervals are long.}
As it is natural to estimate impact over a large sample, we consider $L\rightarrow\infty$ and set
\begin{equation}\label{eq.Iinf.def}
I(Q,\theta):=\lim_{L\rightarrow\infty} I_L(Q,\theta),
\end{equation}
provided the limit is well defined.
Not surprisingly, the large-sample expected impact on midprice turns out to be connected to the stationary distribution of the fundamental price.
We begin with the following technical result.

\begin{lemma}\label{le:ergodicity}
Let us fix an arbitrary $\theta>0$.
Then, there exist unique stationary distributions of $\hat X\,\texttt{mod}\,1$ and $\hat Y\,\texttt{mod}\,1$, with the densities $\psi$ and $\chi$, respectively. These densities are uniquely determined by the following conditions:
\begin{equation}\label{eq.psi.def.cond}
\hat\sigma^2(\theta,\cdot)\psi\in C^{2+\epsilon}([0,1]),\quad \frac{1}{2}\partial^2_x\left(\hat\sigma^2(\theta,x)\psi(x)\right) - \partial_x\left(\hat\mu_0(\theta,x)\psi(x)\right)=0,
\quad x\in(0,1),
\end{equation}
$$
\psi(0^+)=\psi(1^-),\quad \int_0^1 \psi(x)=1,
$$
\begin{equation}\label{eq.chi.def.cond}
\hat\sigma^2(\theta,\cdot)\chi\in C^{2+\epsilon}([0,1]),\quad \frac{1}{2}\partial^2_x\left(\hat\sigma^2(\theta,x)\chi(x)\right) - \partial_x\left(\hat\mu_1(\theta,x)\chi(x)\right)=0,
\quad x\in(0,1),
\end{equation}
$$
\chi(0^+)=\chi(1^-),\quad \int_0^1 \chi(x)=1.
$$
Moreover, for any bounded Borel-measurable function $G$, we have, for any $x\in\RR$,
$$
\lim_{T\rightarrow\infty} \frac{1}{T} \int_0^T \EE\, G(\hat X_t\,\texttt{mod}\,1) dt
= \lim_{T\rightarrow\infty} \EE\, G(\hat X_T\,\texttt{mod}\,1) = \int_0^1 G(z) \psi(z) dz,
$$
$$
\lim_{T\rightarrow\infty} \frac{1}{T} \int_0^T \EE\, G(\hat Y_t(x)\,\texttt{mod}\,1) dt 
= \lim_{T\rightarrow\infty} \EE\, G(\hat Y_T(x)\,\texttt{mod}\,1) = \int_0^1 G(z) \chi(z) dz.
$$
\end{lemma}
\begin{proof}
W.l.o.g. we only consider the case of $\hat X\,\texttt{mod}\,1$. 
First, we notice that the assumptions on $F$ and $\sigma$ imply that $\hat\mu_i(\theta,\cdot)/\hat\sigma^2(\theta,\cdot)\in C^{1+\epsilon}([0,1])$, for $i=0,1$.
Then, for any $c>0$, Theorem 6.5.3 in \cite{Krylov} yields the existence and uniqueness of $\hat\sigma^2\psi\in C^{2+\epsilon}([0,1])$ satisfying the ODE in \eqref{eq.psi.def.cond} with the boundary conditions $\hat\sigma^2(\theta,0^+)\psi(0^+)=\hat\sigma^2(\theta,1^-)\psi(1^-)=c$. Moreover, the maximum principle (or the Feynman-Kac formula) implies that $\psi>0$. Hence, choosing $c>0$ appropriately, we can ensure that $\int_0^1 \psi(x)=1$. Thus, we have shown the existence and uniqueness of the solution to \eqref{eq.psi.def.cond}.

By choosing an arbitrary $f\in C^{2}([0,1])$, satisfying $f(0)=f(1)=0$ and $f'(0^+)=f'(1^-)$, applying It\^o's formula to $f\circ (\cdot\,\texttt{mod}\,1)(\hat X)$, integrating by parts, and using $(\hat\sigma^2\psi)(0^+)=(\hat\sigma^2\psi)(1^-)$, along with the ODE \eqref{eq.psi.def.cond} and the dominated convergence, we show that 
\begin{equation}\label{eq.ergod.pf.1}
\left.\frac{d}{dt}\int_0^1\EE f(\hat X_t(x)\,\texttt{mod}\,1)\psi(x)dx\,\right|_{t=0}=0,
\end{equation}
where $\hat X(x)$ is the solution to \eqref{eq.Xhat.def} with $X_0=x$.
As follows from Theorem 5.4.20 and Remark 5.4.21 in \cite{KaratzasShreve}, $\{\hat X(x)\}_{x\in\RR}$ is a Markov family with the transition denoted by $K(x,A)$. Due to periodicity of the coefficients in \eqref{eq.Xhat.def}, we have $K(x+n,A+n)=K(x,A)$. Then, it is easy to see that $\{\hat X(x)\,\texttt{mod}\,1\}_{x\in[0,1)}$ is a Markov family with the transition kernel $\sum_{n=-\infty}^\infty K(x,A+n)$.
The Markov property and \eqref{eq.ergod.pf.1} imply that $\psi$ is stationary.

To show uniqueness of the stationary distribution of $\hat X\,\texttt{mod}\,1$, consider any such distribution and use the scale function transformation, along with the continuous differentiability and Gaussian estimates for the fundamental solution of a linear (strictly) parabolic PDE with Lipschitz coefficients (see, e.g., \cite{Friedman}), to conclude that the stationary distribution has density $\psi\in C^1([0,1])$, s.t. $\psi(0)=\psi(1)$. Using It\^o's formula, we show that $\psi$ is a weak solution to the ODE in \eqref{eq.psi.def.cond} on $(0,1)$, with the test functions in $C_0^2((0,1))$.
Using the weak form of the ODE \eqref{eq.psi.def.cond}, we improve the regularity and conclude that $\hat\sigma^2(\theta,\cdot)\psi \in C^{2+\epsilon}((0,1))$, and, in turn, that the ODE \eqref{eq.psi.def.cond} holds in classical sense.
Thus, the first part of the proof yields uniqueness of the stationary distribution.

Finally, to obtain the last statement of the lemma, it is a standard exercise to check that the families of measures
$$
\QQ_T(dx):=\frac{1}{T} \int_0^T \PP(\hat X_t\,\texttt{mod}\,1\,\in dx) dt,
\quad \tilde\QQ_T(dx):=\PP(\hat X_T\,\texttt{mod}\,1\,\in dx),
\quad x\in [0,1],
$$
parameterized by $T\geq0$, are tight and that each of their limit points (in weak topology), as $T\rightarrow\infty$, is a stationary distribution of $\hat X\,\texttt{mod}\,1$. Since such distribution is unique, we obtain the statement of the lemma for bounded continuous $G$. As the stationary distribution has no atoms, this statement is extended to all bounded Borel-measurable $G$.
\qed
\end{proof}

\begin{proposition}\label{prop:impact.inf.L}
For any $Q>0$ and $\theta>0$, the limit in \eqref{eq.Iinf.def} is well defined, and we have:
\begin{equation}\label{eq.I0.rep.1}
I(Q,\theta)=\int_0^1 \EE \left(\lceil \hat Y_{Q}(x)\rceil - 1 \right) \psi(x)\,dx,
\end{equation}
where $\psi$ is the density of the stationary distribution of $\hat X\,\texttt{mod}\,1$.
\end{proposition}
\begin{proof}
Notice that, for any $x\in\RR$ and any integer $n$, $\hat Y(x+n)=\hat Y(x)$.
Then, using the independence of $\hat X$, $\hat Y$ and $\eta$, and the uniform distribution of $\eta$, we have:
$$
I_{L}(Q) 
= \EE \left(\lceil \hat Y_Q(\hat X_{\eta})\rceil - \lceil \hat X_{\eta}\rceil\right)
=  \EE \frac{1}{L}\int_0^L \left(\lceil \hat Y_Q(\hat X_{s}\,\texttt{mod}\,1)\rceil - \lceil \hat X_{s}\,\texttt{mod}\,1\rceil\right) ds
= \EE \frac{1}{L}\int_0^L G\left(\hat X_{s}\,\texttt{mod}\,1\right) ds,
$$
where
$$
G(x):= \EE \left(\lceil \hat Y_{Q}(x)\rceil - \lceil x\rceil \right).
$$
Using the ergodicity of $X\,\texttt{mod}\,1$ (see Lemma \ref{le:ergodicity}),
$$
\EE \frac{1}{L}\int_0^L G\left(\hat X_{s}\,\texttt{mod}\,1\right) ds \rightarrow
\int_0^1 G(x) \psi(x) dx.
$$
\qed
\end{proof}

\subsection{Marginal expected impact}
\label{subse:ImpactWings}

In this subsection, we express the marginal expected price impact of a VWAP meta-order (i.e., the derivative of the expected price impact w.r.t. the size $Q$ of the meta-order) through the wings of the stationary distribution of the fundamental price run on the business clock. This connection is one of the key steps in establishing the desired concavity of the expected price impact, as described in Section \ref{se:intro}. The target representation is derived for $Q=0$ (Proposition \ref{prop:marginalimpact.at.zero}) and for $Q=\infty$ (Proposition \ref{prop:Iq.lim}).

\medskip

First, we analyze the asymptotic behavior of $\partial_Q I(Q,\theta)$ as $Q\downarrow0$.

\begin{proposition}\label{prop:marginalimpact.at.zero}
For any $\theta>0$,
$$
\lim_{Q\downarrow0}\partial_{Q}I(Q,\theta) = \alpha\,\psi(1^-).
$$
\end{proposition}
\begin{proof}
Since the drift and volatility of $\hat Y(x)$ are bounded and continuous and the volatility is bounded away from zero, as $Q\rightarrow0$, we have, uniformly over $x\in[0,1]$:
$$
\EE \left(\lceil \hat Y_{Q}(x)\rceil - 1 \right)
= \left(\PP\left(x + \hat\mu_1(\theta,x) Q + \hat\sigma(\theta,x)\sqrt{Q} \hat W_{1} \geq 1\right)\right.
$$
$$
\left.
- \PP\left(x + \hat\mu_1(\theta,x)Q + \hat\sigma(\theta,x)\sqrt{Q} \hat W_{1} \leq 0\right)
\right)(1 + o(1))
$$
$$
=\left( \Phi\left( \frac{x-1+ \hat\mu_1(\theta,x)Q}{\hat\sigma(\theta,x)\sqrt{Q}} \right)
- \Phi\left( -\frac{x+ \hat\mu_1(x)Q}{\hat\sigma(\theta,x)\sqrt{Q}} \right)
\right)(1 + o(1)),
$$
where $\Phi$ is the standard normal c.d.f..
Then, using the continuity of $\psi$, the conditions $\psi(0^+)=\psi(1^-)$ and $\hat\sigma(\theta,x)=\hat\sigma(\theta,1-x)$, for $x\in(0,1)$, as well as the mean value theorem and the dominated convergence, we obtain, as $Q\rightarrow0$:
$$
\int_0^1 \EE \left(\lceil \hat Y_{Q}(x)\rceil - 1 \right) \psi(x)\,dx
$$
$$
= \int_0^1 \left(\Phi\left( \frac{x-1+ \hat\mu_1(\theta,x)Q}{\hat\sigma(\theta,x)\sqrt{Q}} \right)\psi(x) 
-\Phi\left( \frac{x-1-\hat\mu_1(\theta,1-x)Q}{\hat\sigma(\theta,x)\sqrt{Q}} \right) \psi(1-x)\right)\,dx\,(1 + o(1))
$$
$$
= \sqrt{Q} \int_{-1/\sqrt{Q}}^0 \left[\Phi\left( \frac{x}{\hat\sigma(\theta,1+x\sqrt{Q})} + \frac{\hat\mu_1(\theta,1+x\sqrt{Q})\sqrt{Q}}{\hat\sigma(\theta,1+x\sqrt{Q})} \right) \psi(1+x\sqrt{Q})\right.
$$
$$
\left. -\Phi\left( \frac{x}{\hat\sigma(\theta,1+x\sqrt{Q})} - \frac{\hat\mu_1(\theta,-x\sqrt{Q})\sqrt{Q}}{\hat\sigma(\theta,1+x\sqrt{Q})}\right)
\psi(-x\sqrt{Q}) \right]\,dx\,(1 + o(1))
$$
$$
= Q \psi(1^-) \frac{\hat\mu_1(\theta,1^-)+\hat\mu_1(\theta,0^+)}{\hat\sigma(\theta,1^-)} \int_{-\infty}^0 \phi\left( \frac{x}{\hat\sigma(\theta,1^-)}\right) dx\,
(1 + o(1))
$$
$$
= Q \frac{\psi(1^-)(\hat\mu_1(\theta,1^-)+\hat\mu_1(\theta,0^+))}{2} \,(1 + o(1))
= Q \alpha\,\psi(1^-) \,(1 + o(1)).
$$
\qed
\end{proof}


\medskip

Similarly, we can analyze $\partial_Q I(Q,\theta)$ as $Q\rightarrow\infty$. Notice that, due to the Markov property of $\hat Y$ (see analogous argument for the Markov property of $\hat X$ in the proof of Lemma \ref{le:ergodicity}) and the periodicity of the coefficients in \eqref{eq.Y.def}, we have:
\begin{align*}
&I(Q+\Delta Q,\theta)-I(Q,\theta)=\lim_{L\rightarrow\infty} \left(I_L(Q+\Delta Q,\theta)-I_L(Q,\theta)\right)\\
&= \lim_{L\rightarrow\infty} \EE\left(\lceil \hat Y_{Q+\Delta Q}(\hat X_{\eta})\rceil - \lceil \hat Y_{Q}(\hat X_{\eta})\rceil \right)
= \lim_{L\rightarrow\infty} \EE\left(\lceil \hat Y_{\Delta Q}(R_Q(\hat X_{\eta}))\rceil - \lceil R_Q(\hat X_{\eta})\rceil \right)\\
&= \lim_{L\rightarrow\infty} \EE\left(\lceil \hat Y_{\Delta Q}(Z)\rceil - 1 \right),
\end{align*}
where $R(x)\sim \hat Y(x)\,\texttt{mod}\,1$ and $Z\sim R_Q(\hat X_{\eta})\,\texttt{mod}\,1$ are independent of $(\hat X,\hat Y)$.
Repeating the proof of Proposition \ref{prop:impact.inf.L}, we obtain
$$
\lim_{L\rightarrow\infty} \EE\left(\lceil \hat Y_{\Delta Q}(Z)\rceil - 1 \right)
= \int_0^1 \EE\left(\lceil \hat Y_{\Delta Q}(R_{Q}(x))\rceil - 1 \right) \psi(x)dx.
$$

\smallskip

Applying the scale transformation to $\hat Y(x)$, to eliminate the drift, it is easy to see that the density of $\hat Y_t(x)$, denoted $\bar\chi^x_{t}$, can be written as
$$
\bar \chi^x_t(y) = \Gamma(t,x,y)P(y),\quad t>0,\,x,y\in\RR,
$$
where $\Gamma(\cdot,x,\cdot)\in C^{1+\varepsilon,1+\varepsilon}$, with some $\varepsilon\in(0,1)$, is the fundamental solution of the parabolic PDE associated with the transformed $\hat Y^x$, and $P$ is an exponentially bounded Lipschitz-continuous function whose derivative is continuous everywhere except integers, where it has first order discontinuities. A direct computation shows that 
\begin{equation}\label{eq.barchi.PDE}
\partial_t \bar \chi^x - \frac{1}{2}\partial^2_y\left(\hat\sigma^2 \bar \chi^x\right) + \partial_y\left(\hat\mu_1 \bar \chi^x\right)=0,
\end{equation}
where the equation holds globally in $(t,y)\in(0,\infty)\times\RR$ in a weak sense and pointwise (with all derivatives being well defined) everywhere except $(0,\infty)\times \ZZ$, with the left and the right limits being well defined at every integer $y$.
Then, applying the Gaussian estimates for $\Gamma$, it is easy to see that, for any $t>0$, the distribution of $R_{t}(x)\sim \hat Y_{t}(x)\,\texttt{mod}\,1$ has density
$$
\chi^x_{t}(y) = \sum_{n\in\ZZ} \bar \chi^x_t(y+n),\quad y\in[0,1)
$$ 
(a similar argument was used in the proof of Lemma \ref{le:ergodicity}).
Due to periodicity of the coefficients, we deduce that $\chi^x$ satisfies \eqref{eq.barchi.PDE} in the same sense as $\bar \chi^x$.

\smallskip

Repeating the proof of Proposition \ref{prop:marginalimpact.at.zero}, we obtain, as $\Delta Q\rightarrow0$:
$$
\EE\left(\lceil \hat Y_{\Delta Q}(R_{Q}(x))\rceil - 1 \right)
=\int_0^1 \EE \left(\lceil \hat Y_{\Delta Q}(y)\rceil - 1 \right) \chi^x_{Q}(y)\,dy
= \Delta Q \alpha\,\chi^x_{Q}(1^-) \,(1 + o(1)),
$$
for every $x\in(0,1)$. Thus, using the dominated convergence theorem, we conclude:
\begin{equation}\label{eq.IbarV}
\partial_{Q}I(Q,\theta) = \alpha\,\int_0^1 \chi^x_{Q}(1^-) \psi(x) dx.
\end{equation}
The following proposition describes the asymptotic behavior of $\partial_{Q}I$ for large $Q$.

\begin{proposition}
\label{prop:Iq.lim}
For any $\gamma,\theta>0$, 
$$
\lim_{Q\rightarrow\infty}\partial_{Q}I(Q,\theta) = \alpha\,\chi(1^-),
$$
with $\chi$ defined in Lemma \ref{le:ergodicity}.
\end{proposition}
\begin{proof}
Using Ito's formula, it is easy to see that $u(t,y)=\int_0^1\chi^x_t(y)\psi(x) dx$ is a weak solution to
\begin{equation}\label{eq.chi.PDE}
\partial_t u - \frac{1}{2}\partial^2_y\left(\hat\sigma^2 u\right) + \partial_y\left(\hat\mu_1 u\right)=0,
\quad y\in(0,1),\quad u(t,0)=u(t,1)=\int_0^1\chi^x_t(1^-)\psi(x)dx,\quad u(0,y)=\psi(y),
\end{equation}
with $u,\partial_y u \in L^2([0,T]\times[0,1])$.
Recalling that $\chi^x$ satisfies \eqref{eq.barchi.PDE} along $(0,\infty)\times\{1^-\})$, we deduce that $v(t,y):=u(t,y)-\int_0^1\chi^x_t(1^-)\psi(x) dx$ satisfies \eqref{eq.chi.PDE} with the same initial and with zero boundary conditions.
Applying Theorem III.2.1 in \cite{Lady} to $v$, we conclude that $\|\partial_y v(t,\cdot)=\partial_y u(t,\cdot)\|_{L^2}$ is bounded uniformly over $t\geq0$. This, in turn, yields uniform continuity of the family $\{u(t,\cdot)\}_{t\geq0}$. Since the weak limit of this family, as $t\rightarrow\infty$ is $\chi$ (see Lemma \ref{le:ergodicity}), we conclude that it is also a strong limit in the uniform norm. This, along with \eqref{eq.IbarV}, completes the proof of the proposition.
\qed
\end{proof}

\medskip

Thus, we have shown that the marginal expected impact on the midprice at the beginning of an execution sequence is proportional to the stationary distribution of $\hat X \,\texttt{mod}\,1$ at $1^-$. Similarly, we have shown that the marginal expected impact at the end of a sufficiently long execution sequence is proportional to the stationary distribution of $\hat Y\,\texttt{mod}\,1$ at $1^-$. In the next section, we show that, for small $\theta>0$, the former exceeds the latter, which proves the asymptotic concavity of the expected impact curve.

\subsection{Asymptotic concavity of price impact}
\label{subse:concavity}

In this subsection, we show the asymptotic concavity of the expected price impact of a VWAP meta-order: namely, for sufficiently small participation rates $\theta>0$, the marginal expected price impact at the beginning of the meta-order (i.e., for $Q\approx0$) is strictly larger than at its end (i.e., for $Q\approx\infty$). This is done using the representation of the marginal expected price impact via the wings of the stationary distribution, established in the previous subsection. The main result is summarized in Theorem \ref{thm:main}.

\smallskip

In view of Propositions \ref{prop:marginalimpact.at.zero} and \ref{prop:Iq.lim}, the asymptotic concavity of the expected price impact is equivalent to $\psi(1^-)>\chi(1^-)$.
To show the latter, we recall the ODEs \eqref{eq.psi.def.cond} and \eqref{eq.chi.def.cond} and, multiplying them by $\theta$, we deduce the existence and uniqueness of function $(\theta,x)\mapsto f(\theta,x)$, s.t. $f(\theta,\cdot)\in C^{2+\epsilon}([0,1])$ and
\begin{equation}\label{eq.chi.def.cond.2}
\frac{1}{2}\partial^2_x f(\theta,x) - \partial_x\left((\bar\mu_0(x) + \theta \bar\mu_1(x)) f(\theta,x)\right)=0,
\end{equation}
\begin{equation}\label{eq.chi.def.cond.3}
f(\theta,1)=f(\theta,0),\quad \int_0^1 \frac{f(\theta,x)}{\bar\sigma^2(x)}dx = 1,
\end{equation}
where
$$
\bar\sigma(y) := \sqrt{\theta}\hat\sigma(y) = \frac{\sigma(y)}{\sqrt{\gamma (F(y-1)+F(-y))}},
\quad\bar\mu_0(y) = \alpha\gamma\frac{F(y-1) - F(-y)}{\sigma^2(y)},
\quad \bar\mu_1(y):= 2\alpha\gamma \frac{F(-y)}{\sigma^2(y)},
$$
and we used
$$
\beta^+(y) = 1-y,\quad \beta^-(y)=-y,\quad y\in(0,1).
$$
It is clear that $\psi=f(0,\cdot)/\bar\sigma^2$ and $\chi=f(\theta,\cdot)/\bar\sigma^2$. Our goal is to show that, for small enough $\theta>0$, we have $f(\theta,1)<f(0,1)$.
The next proposition establishes the desired result, but under two additional technical assumptions.

\smallskip

\begin{ass}\label{ass:1}
There exists a constant $\rho\geq1$, s.t.
$$
\sigma(x) = \rho \sqrt{\gamma (F(x-1)+F(-x)},\quad x\in(0,1).
$$
\end{ass}

Note that all standing assumptions on $\sigma$ are implied by the above assumption and the properties of $F$.

\begin{ass}\label{ass:2}
The function $F$ is log-concave in $[-1,0]$ and $F'$ is nondecreasing in this range.\footnote{Note that the log-concavity of $F$ can be ensured by requiring that the density of the distribution defined by $F$ is log-concave.} 
\end{ass}

\smallskip

\begin{proposition}\label{prop:ftheta.neg}
For any $x\in[0,1]$, there exists $\partial_\theta f(\cdot,x)\in C(\RR)$. Moreover, under Assumptions \ref{ass:1} and \ref{ass:2}, we have: $\partial_\theta f(0,1)<0$.
\end{proposition}
\begin{proof}
It is easy to see (e.g., using Feynman-Kac formula) that $f(\theta,x)$ is continuously differentiable in $\theta$.
Then, we differentiate \eqref{eq.chi.def.cond.2} and \eqref{eq.chi.def.cond.3} w.r.t. $\theta$ to obtain 
\begin{equation}\label{eq.g.ODE.theta}
\frac{1}{2}g_{xx} - (\bar\mu_0+\theta \bar\mu_1)g_x - (\bar\mu_0'+\theta \bar\mu_1')g = \partial_x(\bar\mu_1 f(0,\cdot)),
\quad g(\theta,0)=g(\theta,1),
\end{equation}
\begin{equation}\label{eq.g.int}
\int_0^1 \frac{g(\theta,x)}{\bar\sigma^2(x)} dx = 0,
\end{equation}
for  $g(\theta,x):=\partial_\theta f(\theta,x)$.

Next, we consider the case $\theta=0$. The PDE \eqref{eq.chi.def.cond.2} and the property $\partial_xf(0,1/2)=\bar\mu_0(1/2)=0$ (which follows from the fact that $f(0,\cdot)$ is symmetric around $x=1/2$) yield:
$$
\partial_x f(0,x) = 2\bar\mu_0(x) f(0,x),\quad x\in(0,1).
$$
Then,
$$
\partial_x(\bar\mu_1 f(0,x))
= \bar\mu_1' f(0,x) + \bar\mu_1 \partial_x f(0,x)
= \left(\bar\mu_1' + 2 \bar\mu_1 \bar\mu_0 \right)f(0,x)
$$
$$
= \frac{\gamma f(0,x)}{\sigma^4} \left(-2\alpha F'(-x)\sigma^2 - 2\alpha F(-x) \partial_x \sigma^2 + 4\alpha\gamma(F(x-1)-F(-x))F(-x)\right)
$$
$$
= \frac{2\alpha\gamma^2 f(0,x)}{\sigma^4} 
\left(-\rho F'(-x)F(x-1)- \rho F'(-x)F(-x) - \rho F(-x)F'(x-1) \right.
$$
$$
\left.
+ \rho F(-x)F'(-x) + 2F(x-1)F(-x) -2F(-x)F(-x)\right)
$$
$$
\leq \frac{2\alpha\gamma^2 f(0,x)}{\sigma^4}
\left(- F'(-x)F(x-1) - F(-x)F'(x-1) + 2F(x-1)F(-x) -2F(-x)F(-x)\right)
$$
$$
= \frac{2\alpha\gamma^2 f(0,x)}{\sigma^4}
\left(-F'(-x)F(x-1) + F(-x)F'(x-1) - 2F(-x)F'(x-1) + 2F(-x)\int_{-x}^{x-1} F'(z)dz\right).
$$
The right hand side of the above is clearly non-positive for $x\in(0,1/2)$. For $x\in(1/2,1)$ its negativity is implied by the monotonicity of $F'$ on $\RR_-$ and by the log-concavity of $F$:
$$
(\log F)'(x-1) < (\log F)'(-x),
$$
$$
F(-x)F'(x-1) < F'(-x)F(x-1).
$$
Recall the ODE for $g(0,\cdot)$:
\begin{equation}\label{eq.g.ODE}
\frac{1}{2}g_{xx} - \bar\mu_0g_x - \bar\mu_0'g = \partial_x(\bar\mu_1 f(0,\cdot)),
\quad g(0,0)=g(0,1).
\end{equation}
The rest of the proof follows from the maximum principle. Indeed, the ODE in \eqref{eq.g.ODE} and the conditions $\partial_x(\bar\mu_1 f(0,\cdot))<0$, $\bar\mu_0'>0$ imply that $g(0,\cdot)$ cannot have a strictly negative minimum in $(0,1)$ (otherwise, the ODE cannot be satisfied at the minimum point of $g(0,\cdot)$). Then, if $g(0,1)\geq0$, we conclude that $g(0,\cdot)\geq0$, which contradicts \eqref{eq.g.int} (the case $g(0,\cdot)\equiv0$ is easily excluded, since $\bar\mu_1 f(0,\cdot)$ cannot be constant). Thus, we conclude that $g(0,1)<0$ and complete the proof of the proposition.
\qed
\end{proof}

\medskip

Thus, we have proved the main mathematical result of this paper.

\begin{theorem}
\label{thm:main}
Under Assumptions \ref{ass:1} and \ref{ass:2}, there exists $\varepsilon>0$, s.t.
$$
\lim_{Q\downarrow0}\partial_Q I(Q,\theta) > \lim_{Q\rightarrow\infty} \partial_Q I(Q,\theta),
$$
for all $\theta\in(0,\varepsilon)$.
\end{theorem}

\section{Empirical analysis}
\label{se:empirical}

\subsection{Numerical example}
\label{subse:numExample}
In this subsection, we present an analytically tractable model from the class described in the previous section. In addition to illustrating the tractability of the proposed setting, our goal herein is to compute numerically the stationary densities, the expected price impact of a VWAP meta-order, and the price resilience curve (Figures \ref{fig:1}--\ref{fig:2}).

\smallskip

Assume, for simplicity, that $\rho=1$ and that $F$ is the c.d.f. of a uniform distribution on $[-a,a]$, for some $a>1$.\footnote{The restriction $a>1$ in order to satisfy the assumptions on $F$ stated at the beginning of Section \ref{subse:finact.jumpdiff.def}. The main purpose of this technical assumption is to ensure that the rates of arrival of the buy an sell orders do not vanish.} Then, all assumptions made in Section \ref{se:math} are satisfied and, for $x\in[0,1]$,
$$
F(x) = \frac{1}{2a}(x+a),\quad F(x-1)+F(-x)=\frac{2a-1}{2a},
\quad \sigma^2(x) = \frac{\gamma(2a-1)}{2a}
$$
$$
\bar\sigma(x) =1,
\quad\bar\mu_0(x) = \frac{\alpha}{2a-1} (2x-1),
\quad \bar\mu_1(x):= \frac{2\alpha}{2a-1} (a-x).
$$
The ODE \eqref{eq.chi.def.cond.2} becomes
\begin{equation*}
\frac{1}{2}\partial^2_x f(\theta,x) - \frac{\alpha}{2a-1}\partial_x\left((2x(1-\theta) + 2\theta a -1) f(\theta,x)\right)=0.
\end{equation*}
To find the general solution of this ODE, we solve:
$$
u_x - \frac{2\alpha}{2a-1}(2x(1-\theta) + 2\theta a -1) u = C_1,
$$
$$
u(x) = g(x) \exp\left( \frac{\alpha}{2(2a-1)(1-\theta)}(2x(1-\theta) + 2\theta a -1)^2\right),
$$
$$
g'(x) = C_1 \exp\left( -\frac{\alpha}{2(2a-1)(1-\theta)}(2x(1-\theta) + 2\theta a -1)^2\right),
$$
$$
g(x) = C_1 \int_{-\infty}^x \exp\left( -\frac{2\alpha(1-\theta)}{2a-1}\left(y + \frac{2\theta a -1}{2(1-\theta)}\right)^2\right) dy + C_2
$$
$$
= C_3 \Phi\left( 2\sqrt{\frac{\alpha(1-\theta)}{2a-1}} \left(x + \frac{2\theta a -1}{2(1-\theta)}\right)\right) + C_2,
$$
$$
u(x) = \exp\left( \frac{\alpha}{2(2a-1)(1-\theta)}(2x(1-\theta) + 2\theta a -1)^2\right)
\left[ C_3 \Phi\left( 2\sqrt{\frac{\alpha(1-\theta)}{2a-1}} \left(x + \frac{2\theta a -1}{2(1-\theta)}\right)\right) + C_2 \right].
$$
The boundary conditions yield
$$
\exp\left( \frac{\alpha}{2(2a-1)(1-\theta)}(2\theta a -1)^2\right)
\left[ C_3 \Phi\left( 2\sqrt{\frac{\alpha(1-\theta)}{2a-1}} \frac{2\theta a -1}{2(1-\theta)}\right) + C_2 \right]
$$
$$
=\exp\left( \frac{\alpha}{2(2a-1)(1-\theta)}(2(1-\theta) + 2\theta a -1)^2\right)
\left[ C_3 \Phi\left( 2\sqrt{\frac{\alpha(1-\theta)}{2a-1}} \left(1 + \frac{2\theta a -1}{2(1-\theta)}\right)\right) + C_2 \right],
$$

$$
C_2=C_3 \frac{\exp\left( \frac{\alpha(1-2\theta + 2\theta a)^2}{2(2a-1)(1-\theta)}\right) \Phi\left( 2\sqrt{\frac{\alpha(1-\theta)}{2a-1}} \left(1 + \frac{2\theta a -1}{2(1-\theta)}\right)\right)
- \exp\left( \frac{\alpha(2\theta a -1)^2}{2(2a-1)(1-\theta)}\right) \Phi\left( 2\sqrt{\frac{\alpha(1-\theta)}{2a-1}} \frac{2\theta a -1}{2(1-\theta)}\right)}
{\exp\left( \frac{\alpha(2\theta a -1)^2}{2(2a-1)(1-\theta)}\right) - \exp\left( \frac{\alpha(1-2\theta + 2\theta a)^2}{2(2a-1)(1-\theta)}\right)}.
$$
Thus, we have
$$
f(\theta,x) = \frac{u(x)}{\int_0^1 u(y) dy},
$$
$$
u(x) = \exp\left( \frac{\alpha(2x(1-\theta) + 2\theta a -1)^2}{2(2a-1)(1-\theta)}\right)
\left[ 
\exp\left( \frac{\alpha(1- 2\theta + 2\theta a)^2}{2(2a-1)(1-\theta)}\right) \Phi\left( \sqrt{\frac{\alpha(1-\theta)}{2a-1}} \frac{1 - 2\theta + 2\theta a}{1-\theta}\right) \right.
$$
$$
\left.
- \exp\left( \frac{\alpha(2\theta a -1)^2}{2(2a-1)(1-\theta)}\right) \Phi\left( \sqrt{\frac{\alpha(1-\theta)}{2a-1}} \frac{2\theta a -1}{1-\theta}\right)
\right.
$$
$$
\left.
+\left( \exp\left( \frac{\alpha(2\theta a -1)^2}{2(2a-1)(1-\theta)}\right) - \exp\left( \frac{\alpha(1- 2\theta + 2\theta a)^2}{2(2a-1)(1-\theta)}\right) \right)
\Phi\left( 2\sqrt{\frac{\alpha(1-\theta)}{2a-1}} \left(x + \frac{2\theta a -1}{2(1-\theta)}\right)\right)\right].
$$
Figure \ref{fig:1} describes the shape of the stationary density $f(\theta,\cdot)$, for various values of $\theta$. We can see that, as predicted by the theoretical results (recall Proposition \ref{prop:ftheta.neg}), the value of the density at the boundary decreases as $\theta$ increases. Moreover, we see that the stationary density becomes more skewed toward the left, as $\theta$ increases. The latter indicates that, during the execution of a meta-order, the fundamental price is more likely to be closer to the best bid than to the best ask price, which is consistent with the phenomenon of ``improving liquidity" discussed in Section \ref{se:intro}. Indeed, if we interpret the fundamental price as the microprice, the fact that it is closer to the best bid price means that the volume of limit orders at the best ask is higher than the volume at the best bid, which implies better liquidity for the executor (assuming that the total volume at the best bid and ask does not change on average).

\medskip

Next, for a fixed participation rate $\theta\in(0,1]$, we compute the expected price impact as a function of executed volume $Q$ (i.e., the expected price trajectory). Recall the formula \eqref{eq.I0.rep.1},
\begin{equation*}
I(Q,\theta)=\int_0^1 \EE \lceil \hat Y^{x}_{Q}\rceil f(0,x)\,dx - 1,
\end{equation*}
where, in the present case,
\begin{equation*}
\hat Y^x_t = x + \int_0^t \hat\mu_1(\theta,\hat Y^x_u) du
+ \frac{1}{\sqrt{\theta}}\hat W_t,
\end{equation*}
$$
\hat\mu_1(\theta,y)= \frac{\alpha}{\theta(2a-1)} \left(2\theta a - 1 + 2(y\,\text{mod}\,1)(1-\theta)\right).
$$
The PDE for $u^{Q,\theta}(t,x):=\EE \lceil \hat Y^{x}_{Q-t}\rceil$ is given by
$$
u^{Q,\theta}_t + \hat\mu_1(\theta,x) u^{Q,\theta}_x + \frac{1}{2\theta} u^{Q,\theta}_{xx} = 0,
\quad t\in[0,Q),\,\,x\in\RR,
\quad u^{Q,\theta}(Q,x) = \lceil x\rceil.
$$
We use the explicit Euler scheme to approximate $u$ and compute the impact curve $I(\cdot,\theta)$ by approximating numerically the integral
$$
I(Q,\theta) = \int_0^1 u^{Q,\theta}(0,x) f(0,x) dx - 1.
$$
The result of this computation is shown in the left part of Figure \ref{fig:2}.

\medskip

Finally, we address the question of price resilience, which measures the expected trajectory of the midprice after a meta-order has been executed. We assume that the execution of the meta-order lasted long enough, so that the process describing the conditional distribution of the fundamental price run on the business time of the executor, $\hat Y$, had entered into its stationary regime before the execution was over. Mathematically, the latter means that, after the execution, the fundamental price run on the business time of the market follows the process
\begin{equation*}
\check Y^x_t = x + \int_0^t \check\mu_1(\check Y^x_u) du
+ \check W_t,
\end{equation*}
where $\check W$ is a Brownian motion, and
$$
\check\mu_1(y) := \frac{\alpha}{2a-1} \left(2(y\,\text{mod}\,1) - 1\right).
$$
Hence, the price resilience is defined as
\begin{equation*}
R(\bar V,\theta):=\int_0^1 \EE \lceil \check Y_{\bar V}\rceil f(\theta,x)\,dx - 1,
\end{equation*}
where $\bar V$ represents the total traded volume in the market. The right part of Figure \ref{fig:2} shows $R(\bar V,\theta)$ as a function of $\bar V$. Note that, since $\theta=0.2$, the range of values of $\bar V$, in the right part of Figure \ref{fig:2}, is chosen to match the range of values of $Q$, in the left part: indeed, the execution of a meta-order of size $Q$ via a VWAP strategy with participation rate $\theta=0.2$ will terminate when the total traded volume becomes $\bar V=Q/\theta=5Q$. It is clear from the right part of Figure \ref{fig:2} that the price resilience is convex and that the expected midprice does not decay to its initial level, which is consistent with the existing theoretical and empirical findings.

\begin{remark}
Figure \ref{fig:2} indicates that the expected price impact of a VWAP meta-order in the proposed model is asymptotically linear, for large $Q$, hence, the concavity becomes less pronounced for large order sizes. This is indeed the case, and it follows from the existence of $\lim_{Q\rightarrow\infty} \partial_Q I(Q,\theta)$, provided the latter is strictly positive. In principle, thisr limit may be zero, but this is excluded by the assumptions on $F$ and $\sigma$ made at the beginning of Section \ref{subse:finact.jumpdiff.def}.
\end{remark}

\begin{remark}
It is important to note that the choice of the model parameters in the above example was dictated purely by the desire to obtain analytic tractability. In this work, we do not address the important question of what model parameters are consistent with the empirically observed dynamics of microprice and order flow, which requires a separate investigation. Once the realistic parameters are found, one can determine the magnitude of the predicted change in the wings of the stationary distribution, during a VWAP meta-order, for realistic participation rates, to see whether the concavity of the expected price impact predicted by the model is significant for practical purposes.
\end{remark}

\subsection{Testing the main model prediction on market data}
\label{subse:marketdata}

The heuristic argument described in Section \ref{se:intro} shows that the concavity of the price impact can be derived from two predictions. The first prediction is that the stationary distribution of the fundamental price modulo tick size, run on the business clock, is U-shaped. And the second prediction states that, during the execution of a VWAP meta-order, the fundamental price (run on the business clock) obtains an additional constant drift in the direction of the order. Although one cannot test empirically the second assumption without having access to the information about meta-orders, it does not seem that this assumption requires any additional justification beyond common sense. Therefore, in this subsection, we test the first assumption, using only publicly available data, without any information about the meta-orders themselves.

\smallskip

The experiment presented here uses the data from NASDAQ exchange obtained via the ITCH protocol, which provides information about every event in the limit order book. An event may be an execution, an addition, or a cancelation, of a limit order, and the associated prices and volumes are either specified directly or can be recovered from prior events (see \cite{JaimungalBook} for a more detailed description of the ITCH protocol). Using this data, one can reverse-engineer the trade volumes of the volumes of limit orders at the first few levels of the limit order book, at the time of every event.\footnote{The author thanks S. Jaimungal for providing this data. 
}
The time interval we use covers November 3--7 of 2014 and includes the tickers CSCO, INTC, MSFT, VOD, LBTYK, LVNTA

\medskip

Following the discussion in Section \ref{se:intro}, we interpret the fundamental price $X$ as the microprice (measured in ticks) and, in turn, $X\texttt{\,mod\,}1$ as the limit order imbalance:
$$
X_t\,\texttt{mod}\,1 \approx \frac{V^b_t}{V^b_t + V^a_t},
$$
where $V^b$ and $V^a$ are the limit order volumes at the best bid and ask respectively, and $t$ is restricted to the time stamps of the events recorded in the database.
Our goal is to estimate the stationary density $f(0,\cdot)$ of $\hat X\,\texttt{mod}\,1$, where $\hat X$ stands for the fundamental price $X$ run on the business clock:
$$
\hat X_t = X_{(V_\cdot)^{-1}(t)},
$$
where $V_t$ is the total traded volume by time $t$.
Due to the ergodicity of $\hat X\,\texttt{mod}\,1$ (see Lemma \ref{le:ergodicity}), for any $0\leq a< b\leq 1$, we have
\begin{align*}
&\int_a^b f(0,y) dy = \lim_{T\rightarrow \infty} \frac{1}{T}\int_0^T \bone_{[a,b]}(\hat X_t) dt
= \lim_{T\rightarrow \infty} \frac{1}{T}\int_0^{(V_\cdot)^{-1}(T)} \bone_{[a,b]}(X_t) dV_t
= \lim_{T\rightarrow \infty} \frac{1}{V_T}\int_0^{T} \bone_{[a,b]}(X_t) dV_t.
\end{align*}
Splitting the interval $[0,1]$ into $10$ intervals $J_1,\ldots,J_{10}$ of length $0.1$, we obtain the following natural approximation
\begin{equation}\label{eq.empirical.StatDens.prelim}
f(0,x) \approx \frac{10}{V_T}\int_0^{T} \bone_{J_k}(X_t) dV_t,\quad x\in J_k,\quad k=1,\ldots,10.
\end{equation}

\medskip

In order to implement \eqref{eq.empirical.StatDens.prelim}, we need to make an important decision on how to interpret the integral in its right hand side.
Recall that the theoretical connection between the marginal expected price impact and the tails of the stationary distribution of $\hat X\,\texttt{mod}\,1$, given by Propositions \ref{prop:marginalimpact.at.zero} and \ref{prop:Iq.lim}, is established in the infinite-activity regime, in which $V_\cdot$ is continuous and, hence, there is no ambiguity in how to understand the integral in the right hand side of \eqref{eq.empirical.StatDens.prelim}. The infinite-activity assumption is important for the theoretical results of this work, as it allows us to ignore the sizes of market orders (more precisely, the joint distribution of these sizes and the fundamental price) in deriving the expression for the marginal expected impact. On the other hand, this assumption is clearly an abstraction, as the actual market orders do not have infinitesimal sizes and, as a result, the process $V_\cdot$ changes by jumps only. Since the process $X_\cdot$ also jumps at the jump times of $V_\cdot$, we need to choose an appropriate (natural) interpretation of the integral in \eqref{eq.empirical.StatDens.prelim}. The following interpretations are used herein.

\begin{enumerate}
\item To explain the first interpretation (and the resulting estimator), let us recall that the purpose of our estimation of the stationary density is to connect it to the expected marginal impact of a VWAP executor. This connection is based on the observation that, by the nature of the VWAP strategy, the executor's trades are mimicking the evolution of the total traded volume process $V$. Hence, the total fraction of market orders executed while the fundamental price is in the interval $J_k$ is equal to the fraction of the executor's trades submitted while $X$ is in $J_k$. It remains to notice that the latter fraction, for $J_k$ close to zero or one, determines the expected marginal impact of the executor. Thus, a natural interpretation of the right hand side of \eqref{eq.empirical.StatDens.prelim} should provide a reasonable approximation to the fraction of executor's trades submitted while $X$ is in $J_k$. Notice that the discontinuity of $V_\cdot$ makes it impossible to follow the VWAP strategy perfectly:\footnote{There is, of course, a more fundamental reason why it is impossible to follow the VWAP perfectly -- it is because the executor does not possess a perfect predictive power for the future traded volume. However, herein, we focus on a different reason.} indeed, the VWAP executor cannot submit her market orders exactly at the same time as the ones submitted by other traders. Assuming that the executor copies a fraction of every market order submitted by other traders (with the size of the fraction determined by her participation rate) and submits her copy right before the associated market order of another trader with probability $w$, and right after it with probability $1-w$ (independent of everything else), we arrive at the estimator
\begin{align}
&f(0,x)\approx \frac{10}{\sum_{t_i\leq T} \Delta V_{t_i}} \sum_{t_i\leq T} \left(w \bone_{J_k}\left(\frac{V^b_{t_i-}}{V^b_{t_i-} + V^a_{t_i-}}\right)
+ (1-w) \bone_{J_k}\left(\frac{V^b_{t_i+}}{V^b_{t_i+} + V^a_{t_i+}}\right)\right) \Delta V_{t_i},\nonumber\\
&x\in J_k,\quad k=1,\ldots,10,\label{eq.empirical.StatDens.weighted}
\end{align}
where $\{t_i\}$ are the arrival times of the market orders, and $V^{a/b}_{t_i-}$ and $V^{a/b}_{t_i+}$ denote, respectively, the limit order volumes right before and right after the $i$th market order.
We refer to \eqref{eq.empirical.StatDens.weighted} as the ``weighted estimator". Note that, from a numerical point of view, this estimator corresponds to approximating the integrand in \eqref{eq.empirical.StatDens.prelim}, at the jump times of the integrator, via a weighted average of its values to the left and to the right of the jump time.
From a financial point of view, the size of $w$ measures how well and how fast the executor can predict the size of the next market order. An executor who is unable or unwilling to predict the individual sizes of external market orders, or who is not fast enough to submit her order ahead of the external one, will have $w=0$. 
For $w=0.5$, we call \eqref{eq.empirical.StatDens.weighted} the ``equally weighted estimator".

\item The next estimator is based on the assumption that the VWAP executor does not try to predict the sizes of individual market orders submitted by other traders, and that she aims to reproduce the desired fraction of the total traded volume by executing trades of a fixed size each at the (dynamically changing) rate at which other market orders arrive.\footnote{Of course, this strategy would not correspond to a perfect VWAP, but, as mentioned before, the latter is impossible anyway.} In this case, \eqref{eq.empirical.StatDens.prelim} changes to
\begin{equation}\label{eq.empirical.StatDens.prelim.uniform}
f(0,x) \approx \frac{10}{\sum_{t_i\leq T} 1}\int_0^{T} \bone_{J_k}(X_t) d \sum_{t_i\leq T}\bone_{[0,t]}(t_i),\quad x\in J_k,\quad k=1,\ldots,10.
\end{equation}
Assuming, as in item 1 above, that the executor aims to submit each trade at the same time as a market order from another trader arrives, and that she is slightly early with probability $w$, and slightly late with probability $1-w$, we obtain the estimator
\begin{align}
&f(0,x)\approx \frac{10}{\sum_{t_i\leq T} 1} \sum_{t_i\leq T} \left(w \bone_{J_k}\left(\frac{V^b_{t_i-}}{V^b_{t_i-} + V^a_{t_i-}}\right)
+ (1-w) \bone_{J_k}\left(\frac{V^b_{t_i+}}{V^b_{t_i+} + V^a_{t_i+}}\right)\right),\nonumber\\
&x\in J_k,\quad k=1,\ldots,10,\label{eq.empirical.StatDens.uniform}
\end{align}
which we call the ``weighted uniform estimator". For $w=0.5$, we call \eqref{eq.empirical.StatDens.uniform} the ``equally weighted uniform estimator".

\item The last estimator proposed herein is based on interpreting the integral in the right hand side of \eqref{eq.empirical.StatDens.prelim} as a limit of the integrals with continuous integrators (in the spirit of the infinite-activity model studied in the preceding sections). Namely, for each discontinuity time $t_i$ of $V_\cdot$ and for any $\varepsilon>0$ small enough, so that $V_\cdot$ is constant in $[t_i-\varepsilon,t_i)$ and in $(t_i,t_i+\varepsilon]$, we replace
\begin{equation*}
\int_{t_i-\varepsilon}^{t_i+\varepsilon} \bone_{J_k}(X_t) dV_t
\end{equation*}
via
\begin{equation}\label{eq.empirical.StatDens.uniform.int.approx}
R^i_{k}:=\int_{0}^{\Delta V_{t_i}} \bone_{J_k}\left( \frac{V^b_{t_i-}}{V^b_{t_i-} + V^a_{t_i-}-v} \right) dv
\quad\text{or}\quad 
\int_{0}^{\Delta V_{t_i}} \bone_{J_k}\left( \frac{V^b_{t_i-}-v}{V^b_{t_i-} + V^a_{t_i-}-v} \right) dv,
\end{equation}
depending on whether $t_i$ is the time of a buy or a sell order, respectively.\footnote{The integrals in \eqref{eq.empirical.StatDens.uniform.int.approx} can be computed in a closed but somewhat cumbersome form.}
Note that this approximation is based on the assumption that the execution of the $i$th market order takes $2\varepsilon$ units of time. However, since the values of the integrals in \eqref{eq.empirical.StatDens.uniform.int.approx} do not depend on $\varepsilon$, one can always make this assumption, provided $\varepsilon$ is smaller than the resolution at which the events are recorded in the database.
As the integral of $\bone_{J_k}(X_t) dV_t$ over the compliment of $\bigcup_i [t_i-\varepsilon,t_i+\varepsilon]$ is zero, the above determines the value of the associated estimator 
\begin{equation}\label{eq.empirical.StatDens.cont}
f(0,x) \approx \frac{10}{\sum_{t_i\leq T} \Delta V_{t_i}} R^i_k,\quad x\in J_k,\quad k=1,\ldots,10,
\end{equation}
which we refer to as the ``continuous estimator".
\end{enumerate}

\medskip

In order to implement the estimators \eqref{eq.empirical.StatDens.weighted}, \eqref{eq.empirical.StatDens.uniform} and \eqref{eq.empirical.StatDens.cont} in a manner that is consistent with the theoretical part of the paper, one needs to restrict the sample of available trades. First of all, only the trades submitted during the NASDAQ normal trading hours (9:30am--4pm Eastern Standard Time) are considered in this empirical analysis. Second, in order to stay as close as possible to the assumption of a constant one-tick bid-ask spread, we further restrict the universe of market orders to those that arrive when the bid-ask spread is equal to one tick and that do not execute any limit orders located in the higher levels of the limit order book (i.e., higher than the first level, measured right before the market order arrives). In addition, if the $i$th market order is a buy and if the bid-ask spread becomes larger than one tick right after this order is executed, we set $V^a_{t_i+}:=0$, and we proceed symmetrically for the sell orders. This restriction ensures that, in the business time intervals included in our analysis, the desired relationship between the fundamental price and the bid and ask prices is preserved, and that the financial rationale described in the first item of the above list remains valid.

\medskip

The estimated stationary density $f(0,\cdot)$ is presented in Figures \ref{fig:3}--\ref{fig:8}, for each estimator and each ticker in our sample.
Figures \ref{fig:3}--\ref{fig:5} show the estimated stationary density over one week of data. They clearly indicate the U-shape property of the density for most tickers, except LBTYK and LVNTA. The reason that the theoretically predicted U-shape property fails for the latter tickers stems from the fact that they are not large-tick stocks: Table \ref{tb:1} confirms that the percentage of traded volume that occurs while the spread is equal to a single tick is significantly smaller for LBTYK and LVNTA than for the other tickers.\footnote{One may also be tempted to explain the lack of the U-shape property of the stationary densities of LBTYK and LVNTA by the relatively low traded volume of these stocks, which may negatively affect the quality of the estimates. However, the traded volume of LBTYK is very close to that of VOD, and the latter ticker has U-shaped stationary density.} Indeed, for a small-tick stock, the microprice does not satisfy the properties of a fundamental price stated in the introduction: a change in the best bid or ask price may not correspond to the microprice crossing the best bid or ask.

Figure \ref{fig:6} indicates that the U-shape property of the stationary density can also be observed on a single-day estimate. However, the presence of daily trends distorts the symmetry of this distribution. Averaging over multiple days helps restore this symmetry.

Figures \ref{fig:7}--\ref{fig:8} show that the U-shape property of the weighted estimate is not very robust w.r.t. the weight $w$:\footnote{We thank the anonymous referee who pointed it out.} on our sample, it is observed for $w\leq 0.6$ (the smaller is $w$ the stronger is the U-pattern) but disappears, and even reverses, as $w$ approaches one. On the other hand, as shown in Figures \ref{fig:7}--\ref{fig:8}, the weighted uniform estimator produces U-shaped stationary densities on our sample, even for large values of $w$. Overall, the empirical analysis supports the prediction of a U-shaped stationary density.

\medskip

To conclude this section, let us explain why it is very challenging to test empirically the second theoretical prediction of this work -- that the wings of the fundamental price distribution decrease during a meta-order -- without using the meta-order data itself. One may be tempted to treat the time periods with pronounced trends in the order flow as being analogous to the execution intervals of meta-orders. If this was a fair analogy, then, in principle, one could estimate the stationary price distribution over such intervals and compare its wings to the stationary distribution estimated over the entire sample. However, there is a subtlety that makes this task significantly more complicated. The theoretical conclusion of Section \ref{subse:concavity} that the wings of the fundamental price distribution decrease during a meta-order is not just a consequence of the fact that the price obtains a drift of a constant sign during a meta-order. Namely, for the conclusion of Section \ref{subse:concavity} to hold it is also important that the VWAP execution strategy does not alter the relationship between the fundamental price and the arrival rate of the market orders.\footnote{Mathematically, it means that, on the business clock, the drift of the price process that is due to the meta-order is constant.} Indeed, a VWAP executor partially hides her activity by submitting her market orders at the same rate as other market participants. As a result, during the VWAP meta-order, the overall arrival rate of the market orders remains the same (symmetric U-shaped) function of the fundamental price (see \eqref{eq.totVolume.infAct}) -- only the balance between the buy and sell market orders changes during the execution. The latter can also be interpreted as the assumption that, on average, the executor's activity remains largely undetected by the other traders, so that they do not change their behavior and, in particular, the relationship between the limit order imbalance and the arrival rate of the market orders remains unchanged. In general, there is no reason why this property would hold during a typical period of positive or negative trend in the order flow, which means that the distribution of the microprice in such time periods may not have the same properties as during a VWAP meta-order. Indeed, the top left panel of Figure \ref{fig:6} shows the stationary distribution of the microprice modulo tick size on a day with a positive trend in the order flow. It is clear that this distribution has more mass concentrated in $[0.5,1]$ than in $[0,0.5]$. On the other hand, Figure \ref{fig:1} shows that the opposite is expected to occur during a buy VWAP meta-order. Thus, in order to test the second theoretical prediction of this work, one needs to find the time intervals during which the order flow has a significant trend, while the relationship between the microprice and the arrival rate of the market orders remains the same as in the overall sample. Needless to say, it is very challenging to detect such intervals, and it presents an interesting topic for further investigation.




\section{Appendix}

In this appendix, we show that there exists a Poisson random measure (PRM) $M$ with the compensator
$$
\mu(dt,dx) = \lambda dt\otimes dF(x),
$$
such that
\begin{align}
&\int_0^t\int_\RR \left(\bone_{\{x\geq \beta^+(\tilde X_{u-})\}} - \bone_{\{x\leq \beta^-(\tilde X_{u-})\}}\right) M(du,dx)\label{eq.appendix.eq1}\\
&= \int_0^t\int_\RR\sum_{j=1}^K \left(\bone_{\{x\geq \beta^+(\tilde X_{u-})\}} + \bone_{\{x\leq \beta^-(\tilde X_{u-})\}}\right) \zeta^j(\tilde X_{u-},u) M^j(du,dx),\nonumber
\end{align}
where $\{M^j,\zeta^j\}$ are described in \eqref{eq.Nj.def}--\eqref{eq.sec2.Pzetaj.def} and
\begin{align*}
&\tilde X_t = X_0 + \alpha\delta \int_0^t\int_\RR\sum_{j=1}^K \left(\bone_{\{x\geq \beta^+(\tilde X_{u-})\}} + \bone_{\{x\leq \beta^-(\tilde X_{u-})\}}\right) \zeta^j(\tilde X_{u-},u) M^j(du,dx) + \int_0^t \sigma(\tilde X_u) d\tilde B_u.
\end{align*}
First, we fix $j$ and introduce the random function
\begin{align*}
&g(u,x):=x\,\bone_{(\beta^-(\tilde X_{u-}),\beta^+(\tilde X_{u-}))}(x) + \eta^+(u)\,\bone_{\{1\}}(\zeta^j(\tilde X_{u-},u))\,\bone_{\RR\setminus(\beta^-(\tilde X_{u-}),\beta^+(\tilde X_{u-}))}(x)\\
&+ \eta^-(u)\,\bone_{\{-1\}}(\zeta^j(\tilde X_{u-},u))\,\bone_{\RR\setminus(\beta^-(\tilde X_{u-}),\beta^+(\tilde X_{u-}))}(x),
\end{align*}
where $\eta^{+}_u$ and $\eta^-_u$ are random variables independent of everything else and distributed according to
$$
\frac{dF(x)}{F(-\beta^+(\tilde X_{u-}))}\,\bone_{[\beta^+(\tilde X_{u-}),\infty)}(x),
\quad \frac{dF(x)}{F(\beta^-(\tilde X_{u-}))}\,\bone_{(-\infty,\beta^-(\tilde X_{u-})]}(x),
$$
respectively.
Let us show that the random measure
$$
\tilde M^j(du,dx):= M^j(du,dx)\circ g(u,\cdot)^{-1}
$$
is a PRM with the same compensator as $M^j$. Indeed, for any predictable random function $\Phi$ (w.r.t. $\mathbb{F}$ which is the completion of the natural filtration of $(\tilde M^j,\tilde B)$), we have\footnote{To address the potential measure-theoretic questions when conditioning on the random elements $\zeta^j$, $\eta^+$, $\eta^-$, we recall that they only need to be defined at the time coordinates of the atoms of $\{M^j\}$.}
\begin{align*}
&\EE\left( \int_s^t \int_\RR \Phi(u,x) \tilde M^j(du,dx)\, \vert\, \mathcal{F}_s\right)
= \EE\left( \int_s^t \int_\RR \Phi(u,g(u,x)) M^j(du,dx)\, \vert\, \mathcal{F}_s\right)\\
&= \EE\left[\EE\left( \int_s^t \int_\RR \Phi(u,g(u,x)) M^j(du,dx)\, \vert\, \zeta^j, \eta^+, \eta^-, \mathcal{F}^{\{M^j\},\tilde B}_s\right)\vert \,\mathcal{F}_s\right]\\
&= \EE\left[\EE\left( \int_s^t \int_\RR \left(\Phi(u,x)\,\bone_{(\beta^-(\tilde X_{u-}),\beta^+(\tilde X_{u-}))}(x) 
+ \Phi(u,e^+(u))\,\bone_{\{1\}}(f(\tilde X_{u-},u))\,\bone_{\RR\setminus(\beta^-(\tilde X_{u-}),\beta^+(\tilde X_{u-}))}(x)\right.\right.\right.\\
&\left.\left.\left.+ \Phi(u,e^-(u))\,\bone_{\{-1\}}(f(\tilde X_{u-},u))\,\bone_{\RR\setminus(\beta^-(\tilde X_{u-}),\beta^+(\tilde X_{u-}))}(x)\right)\lambda^j\, du\, dF(x)\, \vert\, \mathcal{F}^{\{M^j\},\tilde B}_s\right)_{f=\zeta^j,\,e^+=\eta^+,\,e^-=\eta^-}\vert \,\mathcal{F}_s\right]\\
&=\EE\left[\int_s^t \EE \left( \int_{\beta^-(\tilde X_{u-})}^{\beta^+(\tilde X_{u-})} \Phi(u,x) dF(x) 
+ \Phi(u,\eta^+(u))\,\bone_{\{1\}}(\zeta^j(\tilde X_{u-},u)) (1-F(\beta^+(\tilde X_{u-}))+F(\beta^-(\tilde X_{u-})))\right.\right.\\
&\left.\left.+ \Phi(u,\eta^-(u))\,\bone_{\{-1\}}(\zeta^j(\tilde X_{u-},u))(1-F(\beta^+(\tilde X_{u-}))+F(\beta^-(\tilde X_{u-})))\, \vert\, \zeta^j_{[0,u)}, \eta^+_{[0,u)}, \eta^-_{[0,u)}, \mathcal{F}^{\{M^j\},\tilde B}_u\right)\lambda^j\, du\,\vert \,\mathcal{F}_s\right]\\
&= \EE\left[\int_s^t \left( \int_{\beta^-(\tilde X_{u-})}^{\beta^+(\tilde X_{u-})} \Phi(u,x) dF(x)
+ \int_{\beta^+(\tilde X_{u-})}^{\infty} \Phi(u,x)dF(x)
+ \int_{-\infty}^{\beta^+(\tilde X_{u-})} \Phi(u,x)dF(x)\right)\lambda^j\, du\,\vert \,\mathcal{F}_s\right]\\
&= \EE\left[\int_s^t \int_\EE \Phi(u,x) \lambda^j\, du\,dF(x)\,\vert \,\mathcal{F}_s\right].
\end{align*}
It only remains to notice that the right hand side of \eqref{eq.appendix.eq1} is equal to
$$
\int_0^t\int_\RR \left(\bone_{\{x\geq \beta^+(\tilde X_{u-})\}} - \bone_{\{x\leq \beta^-(\tilde X_{u-})\}}\right) \sum_{j=1}^K \tilde M^j(du,dx),
$$
which verifies our claim with $M:=\sum_{j=1}^K \tilde M^j$. Notice also that since the PRM $M$ and the Brownian motion $\tilde B$ are adapted to the same filtration they are independent.


\bibliographystyle{plain}
\bibliography{NonlinPriceImpact}


\begin{table}[hb]
\centering
 \caption{Total traded volume (in number of shares) and the fraction of this volume that is used in the estimate (i.e., of the volume due to the trades that arrive while the bid-ask spread is at one tick and do not eat into the higher levels of the book), for the time period November  3-7, 2014.}\label{tb:1}
\begin{tabular}{|l|c|c|c|c|c|c|} 
\hline
  & CSCO & INTC & MSFT & VOD & LBTYK & LVNTA \\ \hline
Traded volume & 20,848,074 & 25,212,394 & 29,949,217 & 3,955,416 & 4,602,998 & 1,047,510 \\ \hline
\% of traded volume used & 87.25 & 85.23 & 87.03 & 92.09 & 64.49 & 42.73 \\ \hline
\end{tabular}
\end{table}


\begin{figure}
\begin{center}
  \begin{tabular} {cc}
    {
    \includegraphics[width = 0.45\textwidth]{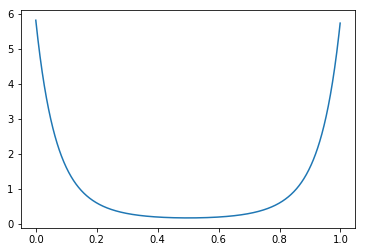}
    } & {
    \includegraphics[width = 0.45\textwidth]{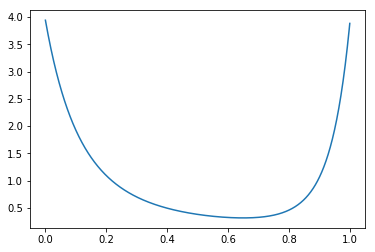}
    }\\
    {
    \includegraphics[width = 0.45\textwidth]{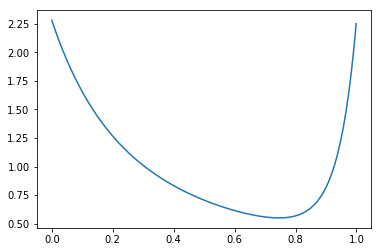}
    } & {
    \includegraphics[width = 0.45\textwidth]{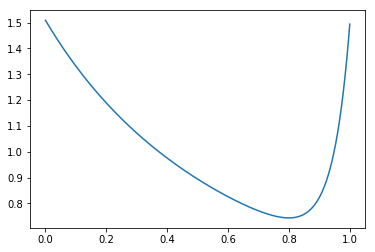}
    }\\
  \end{tabular}
  \caption{Stationary density $f(\theta,\cdot)$, for $\theta=0$ (top left), $\theta=0.2$ (top right), $\theta=0.4$ (bottom left), and $\theta=0.6$. Parameters used are: $F(x)=\bone_{[-a,a]}(x)/(2a)$, $a=1.2$, $\rho=1$, $\alpha=10$.}
    \label{fig:1}
\end{center}
\end{figure}

\begin{figure}
\begin{center}
  \begin{tabular} {cc}
    {
    \includegraphics[width = 0.45\textwidth]{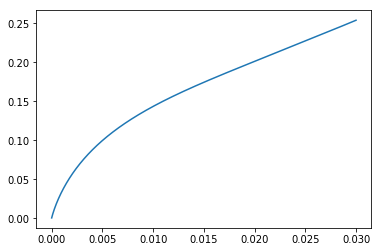}
    } & {
    \includegraphics[width = 0.45\textwidth]{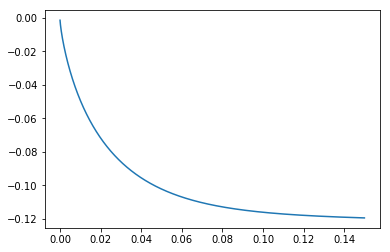}
    }\\
  \end{tabular}
  \caption{Left: the expected impact on midprice $I(Q,\theta)$ (measured in the number of ticks) as a function of executed volume $Q$. Right: the price resilience $R(\bar V,\theta)$ as a function of total traded volume $\bar V$. Parameters used are: $F(x)=\bone_{[-a,a]}(x)/(2a)$, $a=1.2$, $\rho=1$, $\alpha=10$, $\theta = 0.2$.}
    \label{fig:2}
\end{center}
\end{figure}

\begin{figure}
\begin{center}
  \begin{tabular} {cc}
    {
    \includegraphics[width = 0.45\textwidth]{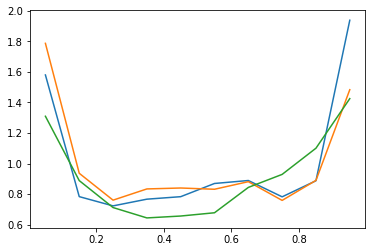}
    } & {
    \includegraphics[width = 0.45\textwidth]{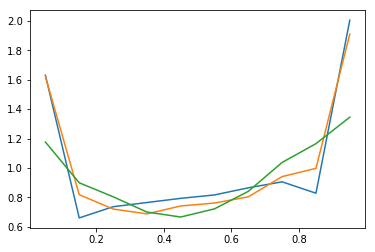}
    }\\
  \end{tabular}
  \caption{Estimated stationary density $f(0,\cdot)$ for CSCO (left) and INTC (right) over November 3--7, 2014. The graphs show the equally weighted estimate (blue line), the equally weighted uniform estimate (orange line), and the continuous estimate (green line).}
    \label{fig:3}
\end{center}
\end{figure}

\begin{figure}
\begin{center}
  \begin{tabular} {cc}
    {
    \includegraphics[width = 0.45\textwidth]{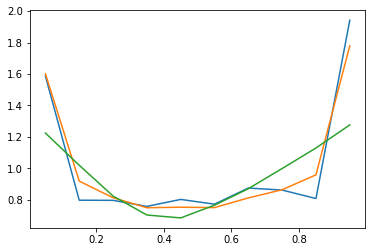}
    } & {
    \includegraphics[width = 0.45\textwidth]{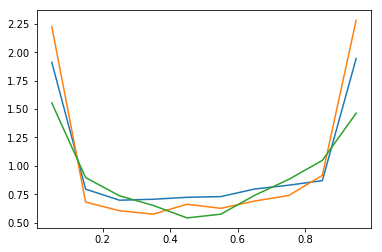}
    }\\
  \end{tabular}
  \caption{Estimated stationary density $f(0,\cdot)$ for MSFT (left) and VOD (right) over November 3--7, 2014. The graphs show the equally weighted estimate (blue line), the equally weighted uniform estimate (orange line), and the continuous estimate (green line).}
    \label{fig:4}
\end{center}
\end{figure}

\begin{figure}
\begin{center}
  \begin{tabular} {cc}
    {
    \includegraphics[width = 0.45\textwidth]{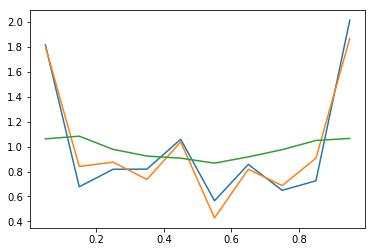}
    } & {
    \includegraphics[width = 0.45\textwidth]{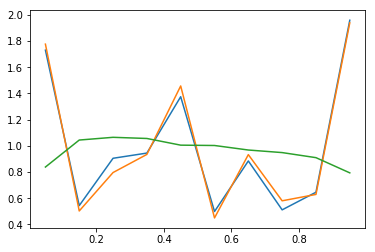}
    }\\
  \end{tabular}
  \caption{Estimated stationary density $f(0,\cdot)$ for LBTYK (left) and LVNTA (right) over November 3--7, 2014. The graphs show the equally weighted estimate (blue line), the equally weighted uniform estimate (orange line), and the continuous estimate (green line).}
    \label{fig:5}
\end{center}
\end{figure}

\begin{figure}
\begin{center}
  \begin{tabular} {cc}
    {
    \includegraphics[width = 0.45\textwidth]{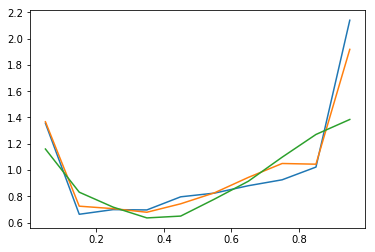}
    } & {
    \includegraphics[width = 0.45\textwidth]{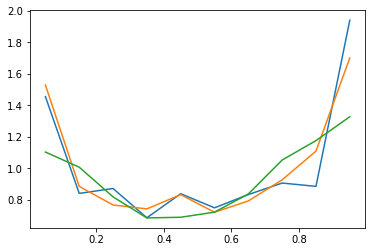}
    }\\
    {
    \includegraphics[width = 0.45\textwidth]{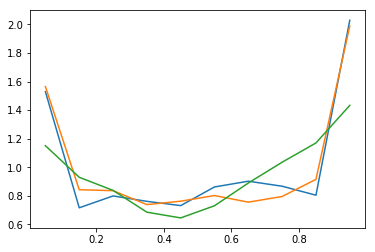}
    } & {
    \includegraphics[width = 0.45\textwidth]{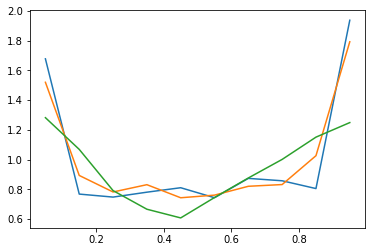}
    }\\
    {
    \includegraphics[width = 0.45\textwidth]{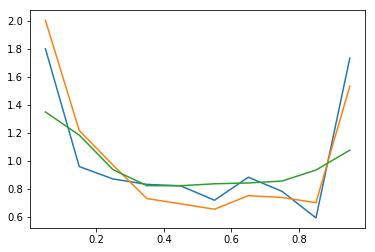}
    } & {
    \includegraphics[width = 0.45\textwidth]{MSFT_imb.png}
    }\\
  \end{tabular}
  \caption{Estimated stationary density $f(0,\cdot)$ for MSFT over: Nov 3, 2014 (top left), Nov 4, 2014 (top right), Nov 5, 2014 (middle left), Nov 6, 2014 (middle right), Nov 7, 2014 (bottom left), Nov 3-7, 2014 (bottom right). The graphs show the equally weighted estimate (blue line), the equally weighted uniform estimate (orange line), and the continuous estimate (green line).}
    \label{fig:6}
\end{center}
\end{figure}

\begin{figure}
\begin{center}
  \begin{tabular} {cc}
    {
    \includegraphics[width = 0.45\textwidth]{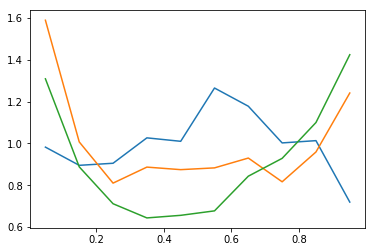}
    } & {
    \includegraphics[width = 0.45\textwidth]{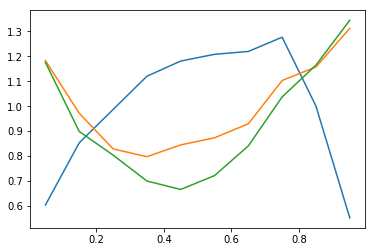}
    }\\
  \end{tabular}
  \caption{Estimated stationary density $f(0,\cdot)$ for CSCO (left) and INTC (right) over November 3--7, 2014. The graphs show the weighted estimate with $w=1$ (blue line), the weighted uniform estimate with $w=1$ (orange line), and the continuous estimate (green line).}
    \label{fig:7}
\end{center}
\end{figure}

\begin{figure}
\begin{center}
  \begin{tabular} {cc}
    {
    \includegraphics[width = 0.45\textwidth]{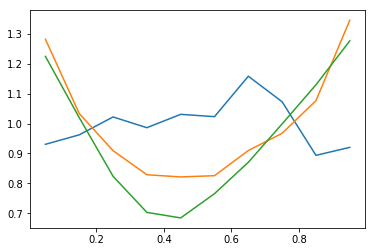}
    } & {
    \includegraphics[width = 0.45\textwidth]{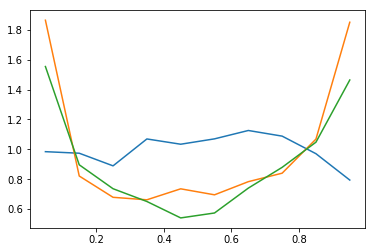}
    }\\
  \end{tabular}
  \caption{Estimated stationary density $f(0,\cdot)$ for MSFT (left) and VOD (right) over November 3--7, 2014. The graphs show the weighted estimate with $w=1$ (blue line), the weighted uniform estimate with $w=1$ (orange line), and the continuous estimate (green line).}
    \label{fig:8}
\end{center}
\end{figure}

\end{document}